\documentclass[11pt,a4paper]{article}
\usepackage[latin1]{inputenc}
\usepackage[T1]{fontenc}
\usepackage{lmodern}
\usepackage[english]{babel}
\normalfont
\usepackage{amsmath}
\usepackage{amssymb}
\usepackage{amsthm}
\usepackage{amsfonts}
\usepackage{dsfont}

\usepackage{mathabx}
\usepackage{bbm}
\usepackage{bm}
\usepackage{mathrsfs}
\usepackage{pifont}
\usepackage{hyperref}
\usepackage{pgf}
\usepackage{graphicx}
\usepackage[inline]{enumitem}
\usepackage{tikz-cd}
\setlist{nosep, itemsep=.1cm, topsep=.1cm}

\newcommand{\be}[0]{\begin{equation}}
\newcommand{\ee}[0]{\end{equation}}

\newcommand*{\textcal}[1]{%
  \textit{\large \fontfamily{pzc}\selectfont#1}%
}

\numberwithin{equation}{section}

\theoremstyle{plain}
\newtheorem{theorem}{Theorem}[section]
\newtheorem{lemma}[theorem]{Lemma}
\newtheorem{proposition}[theorem]{Proposition}

\theoremstyle{definition}

\begin{document}

\vspace*{-1cm}
\thispagestyle{empty}
\vspace*{1.5cm}

\begin{center}
{\Large 
{\bf Hamiltonian anomalies from extended field theories }}
\vspace{2.0cm}

{\large Samuel Monnier}
\vspace*{0.5cm}

Institut für Mathematik,
Universität Zürich,\\
Winterthurerstrasse 190, 8057 Zürich, Switzerland\\
samuel.monnier@gmail.com

\vspace*{1cm}

{\bf Abstract}
\end{center}

We develop a proposal by Freed to see anomalous field theories as relative field theories, namely field theories taking value in a field theory in one dimension higher, the anomaly field theory. We show that when the anomaly field theory is extended down to codimension 2, familiar facts about Hamiltonian anomalies can be naturally recovered, such as the fact that the anomalous symmetry group admits only a projective representation on the state space, or that the latter is really an abelian gerbe rather than an ordinary Hilbert space. We include in the discussion the case of non-invertible anomaly field theories, which is relevant to six-dimensional (2,0) superconformal theories. In this case, we show that the Hamiltonian anomaly is characterized by a degree 2 non-abelian group cohomology class, associated to the non-abelian gerbe playing the role of the state space of the anomalous theory. We construct Dai-Freed theories, governing the anomalies of chiral fermionic theories, and Wess-Zumino theories, governing the anomalies of Wess-Zumino terms, as extended field theories down to codimension 2.

\newpage

\tableofcontents

\section{Introduction and summary}

The Atiyah-Segal axioms \cite{Segal:1987sk, Atiyah1988} picture quantum field theories as functors between a bordism category and the category of Hilbert spaces. A $d$-dimensional quantum field theory indeed assigns a Hilbert space, its space of states, to a $d-1$-dimensional manifold, while path integration on $d$-dimensional bordisms, when such a description is available, provides a homomorphism between the Hilbert spaces associated to the boundaries. The locality of the quantum field theory ensures that this assignment is compatible with the gluing of bordisms. There is in particular a trivial field theory $\bm{1}$ that assigns $1 \in \mathbb{C}$ to any closed $d$-dimensional manifold, $\mathbb{C}$ to any closed $d-1$-dimensional manifolds and trivial homomorphisms to bordisms. 

The extension of these ideas to higher codimension manifolds is known to involve higher category theory \cite{Freed:1994ad, Baez1995, lurie-2009}; going down in dimension requires climbing the higher category hierarchy. For instance, as an extended field theory, the trivial theory $\bm{1}$ assigns the category of (finite dimensional) Hilbert spaces to closed $d-2$-dimensional manifolds. Another extension of these ideas is the notion of \emph{relative} field theory \cite{Freed:2012bs}. Given two $d$-dimensional field theories associated to the same bordism category, a $d-1$-dimensional relative field theory is a natural transformation between the two functors defining the field theories, after the underlying bordism category has been truncated to manifolds of dimension $d-1$ or lower. One can show that a relative theory between two copies of the trivial theory is equivalent to an ordinary field theory. The idea of a relative field theory has a long history, as a "$d-1$-dimensional field theory valued in a $d$-dimensional field theory". It has its roots in Witten's work on the Jones polynomial \cite{Witten:1988hf} and has been implicit in much of the literature on Chern-Simons theory and AdS singletons. 

In a recent paper \cite{Freed:2014iua}, Freed suggested that anomalous field theories should be seen as relative theories between a trivial field theory and a field theory that characterizes the anomaly, the \emph{anomaly field theory}. A similar proposal appeared in unpublished work by Moore \cite{Moore, Moore2012} and in the condensed matter literature \cite{Wen:2013oza, Kong:2014qka}. Further works exploring this idea include \cite{2014arXiv1406.7278F, 2014arXiv1409.5723F}. The first aim of the present paper is to show how many well-known properties of anomalous theories find a natural explanation when one pictures anomalous field theories as relative theories. For many known anomalous field theories, the anomaly field theory is \emph{invertible}, which means in particular that its partition function is non-vanishing and that its state space is one-dimensional. For this reason \cite{Freed:2014iua} focused on the invertible case. The second aim of the present paper is to understand the properties of anomalies associated to non-invertible anomaly field theories. We are mostly interested in the non-invertible case because it is relevant to the six-dimensional (2,0) superconformal field theories. But it is equally relevant to the case 2-dimensional chiral conformal field theories, as we explain briefly in Section \ref{SecExamAnomTh}.

Some properties of anomalous theories are easy to extract from their definition as relative theories. The anomalous theory $\textcal{f}$ is a natural transformation
\be
\label{EqAnThNatTrans}
\textcal{f}: \mathcal{A}|_{d-1} \rightarrow \bm{1}|_{d-1} \;, 
\ee
where $\bm{1}$ and $\mathcal{A}$ are the functors associated respectively to the trivial theory and to the anomaly theory. $|_{d-1}$ denotes the truncation of the bordism category to manifolds and bordism of dimension $d-1$ and lower. \eqref{EqAnThNatTrans} implies that on a $d-1$-dimensional manifold $M^{d-1}$, $\textcal{f}\:(M^{d-1})$ is a unitary complex vector space homomorphism $\mathcal{A}(M^{d-1}) \rightarrow \bm{1}(M^{d-1}) = \mathbb{C}$, hence can be seen as an element of the Hilbert space $\mathcal{A}(M^{d-1})^\dagger = \mathcal{A}(M^{d-1})$. If $\mathcal{A}$ is invertible, $\mathcal{A}(M^{d-1})$ is a Hermitian line and the partition function $\textcal{f}\:(M^{d-1})$ of the anomalous field theory is defined up to a phase. Over the moduli space of $d-1$-dimensional manifolds endowed with appropriate topological/geometric structure, the partition function becomes a section of a possibly non-trivial Hermitian line bundle \cite{Atiyah1984, Freed:1986zx}. If $\mathcal{A}$ is non-invertible, the partition function $\textcal{f}$ is a vector in the Hilbert space $\mathcal{A}(M^{d-1})$. Upon picking a non-canonical identification of $\mathcal{A}(M^{d-1})$ with $\mathbb{C}^n$, we get $n$ $\mathbb{C}$-valued partition functions from the components of this vector. This case is familiar from 2-dimensional rational chiral conformal field theories, which do not have a well-defined partition function, but multiple "conformal blocks", which play the role of partition functions. Over the moduli space, the conformal blocks can be seen as the components of a section of an $n$-dimensional vector bundle. In a completely analogous way, one can deduce that on a $d-2$-dimensional manifold $M^{d-2}$, $\textcal{f}\:(M^{d-2})$ is not quite a Hilbert space. In the invertible case, $\textcal{f}\:(M^{d-2})$ is an object in a category $\mathcal{A}(M^{d-2})^\dagger$ that is equivalent to the category $\mathcal{H}_1$ of Hilbert spaces, but non-canonically so. In the non-invertible case, $\textcal{f}\:(M^{d-2})$ is an object in a category non-canonically equivalent to $\mathcal{H}_1^n$, the $n$th Cartesian product of $\mathcal{H}_1$ with itself. Hence $\textcal{f}\:(M^{d-2})$ can be represented non-canonically as an $n$-component vector of Hilbert spaces. As strange as it may sound, this fact is actually well-known in the case of 2-dimensional rational chiral conformal field theories. The state space of the latter is in general composed of a collection of modules for the chiral vertex algebra underlying the theory. The simple vertex algebra modules can be seen as the generators of a free module category on $\mathcal{H}_1$ that is non-canonically equivalent to $\mathcal{H}_1^n$, where $n$ is the number of simple modules. Over the moduli space, this translates into the fact that the state space of the theory is a possibly non-trivial bundle gerbe \cite{Carey:1995wu, Carey:1997xm}.

However, in the physics literature, anomalies are most often described as the breaking of a classical symmetry under quantization, or more generally as a mild breaking of the invariance under the action of a symmetry group. This description may seem at first sight far removed from Freed's proposal. The key to relate these two points of view lies in the definition of the bordism category. The Atiyah-Segal picture is most often used for topological quantum field theories. In the corresponding bordism category, the unit morphisms are given by cylinders of the form $M^{d-1} \times [0,1]$, for $M^{d-1}$ a $d-1$-dimensional manifold. We are interested here in more general quantum field theories that depend on a set of geometric/topological structures $\mathsf{F}$ that can include a Riemannian metric. Such theories are functors whose domain is a geometric bordism category, composed of manifolds and bordisms carrying an $\mathsf{F}$-structure. In the geometrical realm, there are no bordisms that can play the role of the unit morphisms. Informally, say if the bordisms carry Riemannian metrics, the lengths of the cylinders cannot be ignored and the unit morphisms should be seen as infinitesimal cylinders 
\be
\label{EqLimBord}
\lim_{\epsilon \rightarrow 0} M^{d-1} \times [-\epsilon, \epsilon]
\ee 
of vanishing length. These infinitesimal bordisms must be added by hand in order to obtain a category structure on the set of bordisms \cite{Ayal}. Moreover, we can let the group of automorphisms of the $\mathsf{F}$-structure on $M^{d-1}$, which we write ${\rm Aut}_\mathsf{F}([M^{d-1}])$, act on one end of the cylinders \eqref{EqLimBord}. We also include in the geometric bordism category such infinitesimal bordisms. As a result, the group ${\rm Aut}_\mathsf{F}([M^{d-1}])$ is included in the bordism category, and should be thought of as the (potential) symmetry group of the field theories $\mathcal{A}$ and $\textcal{f}$ on $M^{d-1}$. Writing $\phi$ for the infinitesimal bordism corresponding to the element $\phi \in {\rm Aut}_\mathsf{F}([M^{d-1}])$, $\mathcal{A}(\phi)$ provides a representation of ${\rm Aut}_\mathsf{F}([M^{d-1}])$ on the Hilbert space $\mathcal{A}(M^{d-1})$, while $\bm{1}(\phi)$ corresponds to the trivial representation. As the anomalous field theory is a natural transformation $\textcal{f}: \mathcal{A}|_{d-1} \rightarrow \bm{1}|_{d-1}$, $\textcal{f}\:(M^{d-1})$ is a homomorphism from $\mathcal{A}(M^{d-1})$ to $\mathbb{C}$. Moreover, the definition of a natural transformation implies that $\textcal{f}\:(M^{d-1}) = \textcal{f}\:(\phi M^{d-1}) \circ \mathcal{A}(\phi)$, so $\textcal{f}\:(M^{d-1})$ transforms in the representation of ${\rm Aut}_\mathsf{F}([M^{d-1}])$ dual to the one defined by $\mathcal{A}(\phi)$. If the latter is non-trivial, the theory $\textcal{f}$ is not invariant and the symmetry ${\rm Aut}_\mathsf{F}([M^{d-1}])$ is anomalous.


In the bulk of the paper, we will make precise the sketch of derivation above, in the case where the anomaly field theory $\mathcal{A}$ is extended down to codimension 2 and not necessarily invertible. On $M^{d-1}$, we find that the $n$-dimensional vector space $\mathcal{A}(M^{d-1})$ provides a representation of the anomalous symmetry group ${\rm Aut}_\mathsf{F}([M^{d-1}])$, that can be characterized by a group 1-cocycle of ${\rm Aut}_\mathsf{F}([M^{d-1}])$ valued in $U(n)$. The action of the symmetry group generally permutes the $n$ components of the partition function. This is well-known in the case of 2-dimensional rational chiral conformal field theories: in this case ${\rm Aut}_\mathsf{F}([M^{d-1}])$ contains the "modular transformations", i.e. the large diffeomorphisms of the underlying surface that preserve its conformal structure. Modular transformations generally mix the conformal blocks. In the invertible case, we recover the familiar fact that the partition function (which is unique up to a phase) transforms by multiplication by a character of the anomalous symmetry group ${\rm Aut}_\mathsf{F}([M^{d-1}])$. Going down to $M^{d-2}$, we find in the invertible case that the state space carries only a projective representation of the anomalous symmetry group, characterized by a degree 2 group cohomology class of ${\rm Aut}_\mathsf{F}([M^{d-2}])$ valued in the circle group $\mathbb{T}$. This situation was described by Faddeev \cite{Faddeev:1984jp} and Faddeev-Shatashvili \cite{Faddeev:1985iz} in the 80's. For infinitesimal symmetries, the corresponding degree 2 Lie algebra cocycle was described earlier by Mickelsson in \cite{Mickelsson1985}. Interestingly, the non-invertible case does not seem to have been described in the physics literature before. We find that the vector of Hilbert spaces playing the role of the state space carries something akin to a projective representation of the anomalous symmetry group, but whose projective factors are valued in $S_n \ltimes \mathbb{T}^n$, where $S_n$ is the symmetric group, acting on $\mathbb{T}^n$ by permutation. We show that such a representation naturally yields a degree 2 non-abelian group cocycle \cite{2006math.....11317B} of ${\rm Aut}_\mathsf{F}([M^{d-2}])$ with value in $S_n \ltimes \mathbb{T}^n$. We prove in Propositions \ref{ThHamAnCohomClass} and \ref{ThNonAbCohomClassCharAnom} that the symmetry is anomalous if and only if the corresponding non-abelian cohomology class is non-trivial. This provides a natural extension of the results of Mickelsson-Faddeev-Shatashvili to theories with non-invertible anomalies.

The rest of the paper is dedicated to the construction of invertible anomaly field theories describing well-known anomalies. We construct Wess-Zumino field theories that describe the anomalies of Wess-Zumino terms. Another large class of anomalous field theories are chiral fermions, whose anomalies are described by the Dai-Freed theory \cite{Dai:1994kq}. We extend the construction of Dai and Freed to codimension 2 and perform  some consistency checks, although we do not provide a full proof that the resulting extended field theory is really a 2-functor. The Dai-Freed theory assigns a category constructed from the index gerbe of the corresponding Dirac operator to $d-2$-dimensional manifolds, recovering known results \cite{Segala, Carey:1995wu, Carey:1997xm} about the Hamiltonian anomaly of chiral fermionic field theories.

Note that previous versions of the present paper claimed incorrectly that the anomaly field theories of self-dual fields were of the same type as the Wess-Zumino field theories. The anomaly field theories of uncharged self-dual fields are actually a certain type of Dai-Freed theories. Coupling the self-dual field to a background higher abelian gauge field adds a component to the anomaly field theory involving a Wu Chern-Simons theory \cite{Monnier:2016jlo}. The detailed construction will be presented elsewhere (but see also \cite{Monnier:2017klz, Monniere, MonnierMooreSum2018}).

The present paper is part of a program whose aim is to construct the anomaly field theories of six-dimensional (2,0) superconformal field theories, and in particular to characterize their Hamiltonian anomalies. The partition functions of these anomaly field theories were determined in \cite{Monnier:2014txa}.

The paper is organized as follows. In Section \ref{SecFieldTheo}, we recall the Atiyah-Segal picture of field theories and the notion of extended field theory, focusing on the codimension 2 case of interest to us. Section \ref{SecAnFieldThe} contains the definition of relative field theories and a derivation of some simple consequences. We discuss in detail the consequences of this definition on the symmetries of the theory in dimension $d-1$ and $d-2$. We show that in the non-invertible case, the Hamiltonian anomaly is characterized by a degree 2 non-abelian group cohomology class. Section \ref{SecWZTh} treats Wess-Zumino field theories and Section \ref{SecIndFieldTh} Dai-Freed theories. An appendix contains reviews of several concepts used in the main text. The nLab (\href{http://www.ncatlab.org}{\texttt{http://www.ncatlab.org}}) is a very useful reference for many of the higher categorical concepts appearing in the present paper.

\subsection{Notation}

Here is a brief overview of our notation.
\begin{itemize}
\item Categories, functors and natural transformations, as well as their higher analogues are denoted with calligraphic letters.
\item Objects in categories are denoted by ordinary capitals.
\item Given a 2-category $\mathcal{C}$, its category of morphisms between the objects $X$ and $Y$ is written $\textcal{Hom}_\mathcal{C}(X,Y)$, see Appendix \ref{App2cat2funct}.
\item $M^{d,p}$ is an oriented compact smooth manifold of dimension $d$ with corners down to codimension $p$, $M^d$ is a closed oriented compact smooth manifold of dimension $d$. We will often use this notation to avoid mentioning explicitly the dimension of the corresponding manifold.
\item The disjoint union of manifolds is written as a square cup $\sqcup$.
\item $\mathsf{F}$ denotes a set of geometric/topological structures required to define the quantum field theory of interest. We denote the $\mathsf{F}$-structure on a manifold $M$ by $\mathsf{F}(M)$ and call $M$ an $\mathsf{F}$-manifold. The category of $\mathsf{F}$-manifolds is written $\mathcal{M}_{\mathsf{F}}$.  See Appendix \ref{ApGeomBord2Cat} for further discussion about such structures.
\item In order to kill certain automorphism groups, we sometimes need structures that refine the $\mathsf{F}$-structures used to construct the quantum field theory. We denote those by $\mathsf{E}$. There is a category $\mathcal{M}_{\mathsf{E},\mathsf{F}}$ of manifolds with $\mathsf{E}$-structures whose morphisms only preserve the underlying $\mathsf{F}$-structures. See Appendix \ref{ApGeomBord2Cat}. Given a manifold $M \in \mathcal{M}_{\mathsf{E},\mathsf{F}}$, we will often write $[M]$ for the $\mathsf{F}$-manifold obtained by forgeting the extra data encoded by the $\mathsf{E}$-structure.
\item $\mathcal{B}^{d,p}_\mathsf{F}$ is the bordism $p$-category of $\mathsf{F}$-manifolds of dimension $d$ with $p$-codimensional corners. There is a corresponding bordism category $\mathcal{B}^{d,p}_{\mathsf{E},\mathsf{F}}$ based on the category $\mathcal{M}_{\mathsf{E},\mathsf{F}}$, see Appendix \ref{ApGeomBord2Cat}.
\item Its truncation to manifolds and bordism of dimension $d-1$ or lower is written $\mathcal{B}^{d,p}_\mathsf{F}|_{d-1}$. We use the same notation for the truncation of functors admitting the bordism category as domain, i.e. for field theory functors.
\item For a bordism $M^{d,1}$, we write $\partial_\pm M^{d,1}$ for the outgoing and incoming components of its boundary, so $M^{d,1}: \partial_- M^{d,1} \rightarrow \partial_+ M^{d,1}$.
\item $\mathcal{H}_n$ is the $n$-category of $n$-Hilbert spaces. We will be interested only in the case $n = 0,1,2$. $\mathcal{H}_0$ is the set $\mathbb{C}$. $\mathcal{H}_1$ is the category of finite dimensional Hilbert spaces. $\mathcal{H}_2$ is the 2-category of 2-Hilbert spaces defined in Appendix \ref{App2Vect}.
\item $\mathcal{T}_n$ are the higher circle groups, with $\mathcal{T}_0 = \mathbb{T}$ being the circle group $U(1)$. See Appendix \ref{AppCircGerbes}.
\item Chain, cochains, cycles and cocycles are represented with hats, the corresponding cohomology classes carry no hats. Differential cocycles carry a caron. See Appendix \ref{AppDiffCoc}.
\end{itemize}

\section{Field theories}

\label{SecFieldTheo}

In this section, we introduce some notation and sketch the picture of field theories as functors from a (higher) cobordism category to a (higher) category of (higher) Hilbert spaces. We stay concise, and we refer the reader to Section 1 of \cite{lurie-2009} for a more detailed exposition of extended field theories, in the case of topological field theories.

\subsection{The functorial picture of field theories}

A $d$-dimensional quantum field theory can be thought of as an assignment of a complex number, the partition function, to each $d$-dimensional manifold, and of a Hilbert space, the space of quantum states, to each $d-1$-dimensional manifold. A "manifold" should be understood here as a smooth orientable manifold endowed with all the extra structures required to define the quantum field theory of interest, e.g. a spin structure, a Riemannian or Lorentzian metric, and so on. We will denote this topological and/or geometrical structure by $\mathsf{F}$, and sometimes call a manifold endowed with an $\mathsf{F}$-structure an $\mathsf{F}$-manifold. In addition, path integration over manifolds with incoming and outgoing boundaries provides linear maps between the Hilbert spaces associated to the boundaries. These maps must be compatible with the gluing of manifolds along their boundaries.

The discussion above can be formalized using categorical concepts. There is a bordism category $\mathcal{B}^{d,1}_\mathsf{F}$ of $\mathsf{F}$-manifolds defined as follows (see Appendix \ref{ApGeomBord2Cat} for more details). The objects are $d-1$-dimensional $\mathsf{F}$-manifolds $M^{d-1}$ endowed with the germ of an $\mathsf{F}$-structure on $M^{d-1} \times \{0\} \subset M^{d-1} \times (-\epsilon, \epsilon)$, where $\epsilon > 0$. The morphisms between objects $M_-^{d-1}$ and $M_+^{d-1}$ in $\mathcal{B}^{d,1}_\mathsf{F}$ are $d$-dimensional $\mathsf{F}$-manifolds with boundary $-M_-^{d-1} \sqcup M_+^{d-1}$ extending the germ of $\mathsf{F}$-structure existing on the boundary. $-M_-^{d-1}$ denotes here $M_-^{d-1}$ with its opposite orientation. The composition of morphisms is given by gluing along the boundaries and the germs ensure that smooth $\mathsf{F}$-structures are obtained from the gluing of smooth $\mathsf{F}$-structures. $\mathcal{B}^{d,1}_\mathsf{F}$ admits a symmetric monoidal structure (i.e. a "commutative product") given by the disjoint union of manifolds. It also admits a $\dagger$-category structure, where the $\dagger$ operation is given by inverting the orientation of the bordisms.


There is a category $\mathcal{H}_1$ whose objects are finite dimensional Hilbert spaces and whose morphisms are homomorphisms. The tensor product provides as well a symmetric monoidal structure. $\mathcal{H}_1$ carries a $\dagger$-structure, given by the Hermitian conjugation of homomorphisms. In order to describe most field theories, one may rather want to consider a larger category consisting of infinite-dimensional Hilbert spaces or topological vector spaces and continuous homomorphisms, endowed with a completed tensor product (see for instance Lecture 3 of \cite{Segal}). We will see nevertheless in Section \ref{SecDef} that working with finite-dimensional vector spaces is sufficient to describe most anomalous field theories of physical interest, despite the fact that their state spaces are infinite-dimensional.

A (unitary) quantum field theory is then seen as a functor $\mathcal{F}: \mathcal{B}^{d,1}_{\mathsf{F}} \rightarrow \mathcal{V}_1$ compatible with the monoidal structures (i.e it is symmetric monoidal) and with the $\dagger$-structures. The requirement that $\mathcal{F}$ is a functor ensures that the assignment of homomorphisms of Hilbert spaces to manifolds with boundaries by the quantum field theory is compatible with gluing. The compatibility with the monoidal structure ensures that the partition function on disjoint unions of $d$-dimensional manifolds is the product of the partition functions associated to each connected component. Similarly, it ensures that the Hilbert space/homomorphism associated to disjoint unions of $d-1$-dimensional manifolds/$d$-dimensional bordisms is the tensor product of the Hilbert spaces/homomorphisms associated to the connected components. 
The compatibility with the $\dagger$-structure essentially implements the CPT theorem, known to hold for all unitary quantum field theories.

As any manifold can be seen as the disjoint union of itself and the empty manifold, the compatibility with the monoidal structure requires that $\mathcal{F}(\emptyset^d) = 1$, $\mathcal{F}(\emptyset^{d-1}) = \mathbb{C}$, where we respectively considered the empty set as a $d$-dimensional manifold and as a $d-1$-dimensional manifold. This fact also explains why we can see the field theory as associating a complex number to a closed $d$-dimensional manifold $M^d$. $M^d$ should really be seen as a bordism between $\emptyset^{d-1}$ and itself, which corresponds to a homomorphism $\mathcal{F}(M^d): \mathbb{C} \rightarrow \mathbb{C}$. But the space of such homomorphisms can be canonically identified with $\mathbb{C}$.

\subsection{Extended field theories}

The locality of quantum field theory suggests that one should be able to reconstruct the theory on any manifold $M$ from the knowledge of the theory on elementary building blocks, for instance simplexes. To do this, we must extend the cobordism category to include manifolds of dimension $d' \leq d$ with corners of arbitrary codimension, and understand the kind of objects that the functor $\mathcal{F}$ associates to them. Such quantum field theories are called \emph{fully extended}. It is also well-known that quantum field theories of physical interest often contain, in addition to point-like operators, defect operators of all codimensions. A proper description of such operators would probably also require formulating the theory as a fully extended field theory. We refer the reader to \cite{lurie-2009} for an account of fully extended topological field theories.

Less ambitiously, one may fix some $q < d$ and consider extended theories involving manifolds of dimension $d'$, $q \leq d' \leq d$, with $q$-dimensional corners. As we will see, anomalous field theories in $d-1$ dimensions are related to anomaly field theories in $d$ dimensions, so in order to understand the effect of anomalies on the state spaces of the anomalous theory on $d-2$-dimensional manifolds, we must consider extended anomaly field theories with $q = d-2$.

The definition of the extended bordism category $\mathcal{B}^{d,2}_{\mathsf{F}}$ can be found in Appendix \ref{ApGeomBord2Cat}. In summary, $\mathcal{B}^{d,2}_{\mathsf{F}}$ is a strict 2-category with the following properties. An object in $\mathcal{B}^{d,2}_{\mathsf{F}}$ is a closed manifold $M^{d-2}$ endowed with a $d$-dimensional germ of $\mathsf{F}$-structure. A 1-morphism between $M_-^{d-2}$ and $M_+^{d-2}$ is an manifold $M^{d-1,1}$ with boundary $-M_-^{d-2} \sqcup M_+^{d-2}$ and endowed with a $d$-dimensional germ of $\mathsf{F}$-structure, which should be compatible with the germs existing on the boundary. Such 1-morphisms are called \emph{regular}. There are in addition \emph{limit} 1-morphism, corresponding to infinitesimal bordisms, which will be reviewed later. A 2-morphism between 1-morphisms $M^{d-1,1}_-, M^{d-1,1}_+: M_-^{d-2} \rightarrow M_+^{d-2}$ is a manifold $M^{d,2}$ with boundary is $-M^{d-1,1}_- \sqcup_{-M_-^{d-2} \sqcup M_+^{d-2}} M^{d-1,1}_+$ and corners $-M_-^{d-2} \sqcup M_+^{d-2}$, see \eqref{EqBound2-Morph}. 
A symmetric monoidal structure on $\mathcal{B}^{d,2}_{\mathsf{F}}$ is provided by the disjoint union, and a dagger structure is provided by the orientation reversal.

The target of a field theory extended down to codimension 2 is the 2-category $\mathcal{H}_2$ of 2-Hilbert spaces \cite{MR1278735, Freed:1994ad, 1996q.alg.....9018B}, a notion that we review in Appendix \ref{App2Vect}. In short, a complex vector space is a $\mathbb{C}$-module. Going up in the category hierarchy, the role of $\mathbb{C}$ is taken by the category $\mathcal{V}_1$ of vector spaces, which can be seen as a semiring under the operations of direct sum and tensor product. (In order to get a true ring with an invertible addition, we would need to consider virtual vector spaces, which we will not do.) A 2-vector space is therefore a $\mathbb{C}$-linear category that is also a finitely generated free module for the category of vector spaces, up to equivalence. The simplest 2-vector space, playing a role equivalent to $\mathbb{C}$ for complex vector spaces, is the category $\mathcal{V}_1$ of vector spaces itself. Morphisms of 2-vector spaces are provided by $\mathcal{V}_1$-linear functors, i.e. functors preserving the $\mathcal{V}_1$-module structure, and 2-morphisms are natural transformations. The 2-vector spaces form in this way a 2-category $\mathcal{V}_2$.  $\mathcal{V}_2$ can be endowed with a higher direct sum and higher tensor product operations, forming a semiring structure, with $\mathcal{V}_1$ being the unit for the higher tensor product operation. 

Passing to Hilbert spaces, we need a sesquilinear form valued in $\mathcal{H}_1$ on our 2-vector space. This sesquilinear form is played by the hom functor, so we need to restrict to 2-vector spaces that are enriched in $\mathcal{H}_1$, i.e. whose spaces of morphism between any two objects are Hilbert spaces. The requirement of sesquilinearity requires furthermore that the 2-Hilbert spaces be $H^\ast$-categories \cite{1996q.alg.....9018B}, which are essentially $\dagger$-categories whose dagger operation is compatible with the inner product on the spaces of morphisms, see Appendix \ref{App2Vect}. We obtain in this way a 2-category $\mathcal{H}_2$ of 2-Hilbert spaces. There is a monoidal structure given by the higher tensor product and a dagger structure. The role of $\mathbb{C}$ as the trivial Hilbert space in $\mathcal{H}_1$ is taken over by the trivial 2-Hilbert space $\mathcal{H}_1$ in $\mathcal{H}_2$. We refer the reader to Appendix \ref{App2Vect} for more detailed information.

A field theory with data $\mathsf{F}$ extended to codimension 2 is a 2-functor $\mathcal{F}: \mathcal{B}^{d,2}_{\mathsf{F}} \rightarrow \mathcal{H}_2$ compatible with the monoidal and dagger structures. The functorial property ensures consistency with the gluing of manifolds. The compatibility with the dagger structure implements the CPT theorem and the compatibility with the monoidal structure implements the multiplicative property of the field theory data on disjoint manifolds. As before, the latter puts constraints on the value of $\mathcal{F}$ on the empty set: $\mathcal{F}(\emptyset^d) = 1$, $\mathcal{F}(\emptyset^{d-1}) = \mathbb{C}$, $\mathcal{F}(\emptyset^{d-2}) = \mathcal{H}_1$. This also allows us to simplify our picture of the objects the field theory associates to closed manifolds. For instance a closed manifold $M^{d-1}$ should be seen as a bordism from $\emptyset^{d-2}$ to itself, so $\mathcal{F}$ should associate to it a functor $\mathcal{F}(M^{d-2}): \mathcal{H}_1 \rightarrow \mathcal{H}_1$ preserving the semiring structure on $\mathcal{H}_1$. But any such functor is of the form $\bullet \otimes H$ for some $H \in \mathcal{H}_1$ \cite{2008arXiv0812.4969B}, so we can naturally see $\mathcal{F}(M^{d-1})$ as a Hilbert space, the space of quantum states of the theory. Similarly, one can show that $\mathcal{F}(M^d)$ can be canonically identified with a complex number, the partition function of the theory on $M^d$.

We say that a field theory $\mathcal{F}$ is \emph{invertible} when $\mathcal{F}(M)$ is invertible for all $M$. More precisely, $\mathcal{F}(M^d)$ should be a non-zero complex number, which is obviously invertible with respect to the monoidal structure on $\mathbb{C}$, namely the complex multiplication. $\mathcal{F}(M^{d-1})$ should be a Hermitian line. Hermitian lines are indeed the invertible objects of $\mathcal{H}_1$ with respect to the monoidal structure given by the tensor product. For the same reason, $\mathcal{F}(M^{d-2})$ should be a 2-Hermitian line (see Appendix \ref{App2Vect}). $\mathcal{F}(M^{d,1})$ should be a vector space isomorphism. $\mathcal{F}(M^{d-1,1}): \mathcal{F}(\partial_- M^{d-1,1}) \rightarrow \mathcal{F}(\partial_+ M^{d-1,1})$ should be an invertible functor, in the sense that there is a functor $\mathcal{G}: \mathcal{F}(\partial_+ M^{d-1,1}) \rightarrow \mathcal{F}(\partial_- M^{d-1,1})$ such that the two compositions of $\mathcal{F}(M^{d-1,1})$ with $\mathcal{G}$ are the identity functors on $\mathcal{F}(\partial_- M^{d-1,1})$ and $\mathcal{F}(\partial_+ M^{d-1,1})$. Finally, for $M^{d-2}$ such that $\partial M^{d-2} = -N_- \cup N_+$ with $\partial N_{\pm} = N_+ \cap N_- = -Q_- \sqcup Q_+$, $\mathcal{F}(M^{d-2})$ should be a natural equivalence between the functors $\mathcal{F}(N_-)$ and $\mathcal{F}(N_+)$, which map $\mathcal{F}(Q_-)$ to $\mathcal{F}(Q_+)$.

A trivial example of an extended invertible field theory is the following. Consider the field theory $\bm{1}$ that associates
\begin{itemize}
\item $1$ to any $d$-dimensional manifold.
\item $\mathbb{C}$ to any $d-1$-dimensional manifold.
\item $\mathcal{H}_1$ to any $d-2$-dimensional manifold. 
\end{itemize}
One should interpret the statements above properly in order to reconstruct the corresponding functor. For instance, $\bm{1}(M^{d,1})$ is a homomorphism $\mathbb{C} \stackrel{\bullet \cdot 1}{\rightarrow} \mathbb{C}$, i.e. the identity homomorphism. $\bm{1}(M^{d-1,1})$ is the functor $\mathcal{I}: \mathcal{H}_1 \stackrel{\bullet \otimes \mathbb{C}}{\rightarrow} \mathcal{H}_1$, which is just the identity functor. $\bm{1}(M^{d,2})$ can be identified with the natural transformation $\mathcal{I} \stackrel{\bullet \cdot 1}{\rightarrow} \mathcal{I}$ between the identity functors, i.e. the identity natural transformation. We will see the use of the trivial field theory next.

\section{Anomalous field theories}

\label{SecAnFieldThe}

In this section, we explain how anomalies of quantum field theories can be pictured elegantly using the formalism of extended field theories \cite{Freed:2014iua, Moore, Moore2012}. We develop this formalism to include Hamiltonian anomalies. We also generalize it in order to accommodate non-invertible anomaly field theories, which is the case relevant to the six-dimensional (2,0) theories. We show that in this case, the Hamiltonian anomaly on a spacial slice $M^{d-2}$ is characterized by a non-abelian cohomology class of the automorphism group ${\rm Aut}_{\mathsf{F}}([M^{d-2}])$. More information about relative field theories and their relations to anomalies can be found in \cite{Freed:2012bs, Freed:2014iua, Moore, Moore2012, 2014arXiv1406.7278F, 2014arXiv1409.5723F}.

\subsection{Definition}

\label{SecDef}

We start with a technical remark. As will be obvious in the present section, the bordism category $\mathcal{B}^{d,2}_{\mathsf{F}}$, based on the category $\mathcal{M}_\mathsf{F}$ of $\mathsf{F}$-manifolds, is not appropriate to describe anomalous field theories on $\mathsf{F}$-manifolds. Rather, it will be userful to consider slightly more general bordism categories, based on the categories of manifolds $\mathcal{M}_{\mathsf{E},\mathsf{F}}$ defined in Appendix \ref{ApGeomBord2Cat}. The manifolds in $\mathcal{M}_{\mathsf{E},\mathsf{F}}$ carry a structure $\mathsf{E}$ that refines the structure $\mathsf{F}$, in the sense that any $\mathsf{E}$-structure determines an $\mathsf{F}$-structure. The morphisms in $\mathcal{M}_{\mathsf{E},\mathsf{F}}$ preserve however only the $\mathsf{F}$-structures. We will call such manifolds $(\mathsf{E},\mathsf{F})$-manifolds. The $\mathsf{E}$-structure should be seen as an additional label attached to manifolds with $\mathsf{F}$-structure, which reduces or kill their automorphism groups. Only the $\mathsf{F}$-structure is used to define the quantum field theories involved. We write $\mathcal{B}^{d,1}_{\mathsf{E},\mathsf{F}}$ for the bordism category constructed on $\mathcal{M}_{\mathsf{E},\mathsf{F}}$ and defined in Appendix \ref{ApGeomBord2Cat}. As explained there, any extended field theory on $\mathsf{F}$-manifolds, represented by a functor $\mathcal{F}: \mathcal{B}^{d,2}_{\mathsf{F}} \rightarrow \mathcal{H}_2$, gives rise canonically to a field theory on $(\mathsf{E},\mathsf{F})$-manifold by pulling back the functor $\mathcal{F}$ to a functor $\mathcal{F}': \mathcal{B}^{d,2}_{\mathsf{E},\mathsf{F}} \rightarrow \mathcal{H}_2$.

An interesting fact is that a non-extended $d-1$-dimensional field theory can be seen as a 2-natural transformation  $\textcal{f}: \bm{1}|_{d-1} \rightarrow \bm{1}|_{d-1}$. Here, $\bm{1}$ is seen as a $d$-dimensional extended field theory, and $\bm{1}|_{d-1}$ is the restriction of $\bm{1}$ to the bordism category truncated to manifolds and bordisms of dimension $d-1$ or lower, written $\mathcal{B}^{d,2}_{\mathsf{E},\mathsf{F}}|_{d-1}$. This truncation is discussed at the end of Appendix \ref{ApGeomBord2Cat}. A sketch of a definition of 2-natural transformations can be found in Appendix \ref{App2cat2funct}, see \cite{1998math.....10017L} for a full definition. To see what this means, remark that such a 2-natural transformation associates an element of $\textcal{Hom}_{\mathcal{H}_2}(\bm{1}(M^{d-2}), \bm{1}(M^{d-2})) = \textcal{Hom}_{\mathcal{H}_2}(\mathcal{H}_1, \mathcal{H}_1)$, i.e. an $\mathcal{H}_1$-linear functor $\mathcal{H}_1 \rightarrow \mathcal{H}_1$, to each object $M^{d-2}$ of $\mathcal{B}^{d,2}_{\mathsf{E},\mathsf{F}}$. But we saw that such functors can be represented as the tensor product with a Hilbert space $H \in \mathcal{H}_1$, namely the image of $\mathbb{C} \in \mathcal{H}_1$. $\textcal{f}$ therefore associates a Hilbert space to each closed $d-2$-dimensional manifold $M^{d-2}$. This Hilbert space can be identified with the state space $\mathcal{F}(M^{d-2})$ of a $d-1$-dimensional field theory $\mathcal{F}$. Moreover, $\textcal{f}$ takes a bordism $M^{d-1,1}$ with $\partial M^{d-1,1} = -\partial_- M^{d-1,1} \sqcup \partial_+ M^{d-1,1}$ to a morphism $\textcal{f}\:(M^{d-1,1})$ of the category $\textcal{Hom}_{\mathcal{H}_2}(\mathcal{H}_1, \mathcal{H}_1)$ between $\textcal{f}\:(\partial_- M^{d-1,1})$ and $\textcal{f}\:(\partial_+ M^{d-1,1})$, i.e. a natural transformation between the associated functors. This natural transformation can be pictured as a homomorphism between the corresponding Hilbert spaces $\mathcal{F}(\partial_- M^{d-1,1})$ and $\mathcal{F}(\partial_+ M^{d-1,1})$ \cite{2008arXiv0812.4969B}. The compatibility of $\textcal{f}$ with the gluing in $\mathcal{B}^{d,2}_{\mathsf{E},\mathsf{F}}$ implies that $\mathcal{F}: \mathcal{B}^{d-1,1}_{\mathsf{E},\mathsf{F}} \rightarrow \mathcal{H}_1$ is a functor. If $\textcal{f}$ is required to be compatible with the monoidal and dagger structures, then $\mathcal{F}$ is a (non-extended) field theory in $d-1$ dimensions. The truncation is required, because the compatibility of the 2-natural transformation $\textcal{f}$ with the $d$-dimensional bordisms would require the partition function $\textcal{f}(M^{d-1})$ to be a cobordism invariant, which is in general not the case for quantum field theories of interest.

This suggests that given an extended $d$-dimensional field theory $\mathcal{A}$, we might obtain an interesting generalization of a field theory by looking at 2-natural transformations of the form $\bm{1}|_{d-1} \rightarrow \mathcal{A}|_{d-1}$ or $\mathcal{A}|_{d-1} \rightarrow \bm{1}|_{d-1}$. In fact, the two possibilities are not really different, because, at least in the finite dimensional setting that we are considering here, one can always find a dual field theory $\mathcal{A}^\dagger$ such that there is an equivalence between the 2-natural transformations $\mathcal{A}|_{d-1} \rightarrow \bm{1}|_{d-1}$ and $\bm{1}|_{d-1} \rightarrow \mathcal{A}^\dagger|_{d-1}$. $\mathcal{A}^\dagger$ is obtained from $\mathcal{A}$ by postcomposing it with the dagger operation on $\mathcal{H}_2$. Let us therefore define an \emph{anomalous field theory} to be a 2-natural transformation $\textcal{f}: \mathcal{A}|_{d-1} \rightarrow \bm{1}|_{d-1}$ compatible with the monoidal and dagger structures, $\mathcal{A}$ being the \emph{anomaly field theory} of $\textcal{f}$. We will see momentarily how to recover the physical notion of an anomalous field theory from this definition. Up to the operation of taking the dual, this definition corresponds to what was defined as a relative field theory in \cite{Freed:2012bs}. In \cite{Freed:2014iua} anomalous field theories were defined as relative field theories with the extra requirement that the anomaly field theory $\mathcal{A}$ should be invertible. We find it suitable to broaden the definition of \cite{Freed:2014iua} in order to accommodate the chiral rational conformal field theories in two dimensions, or the six-dimensional (2,0) superconformal field theories.

\paragraph{Properties of anomalous theories} Let us try to understand the consequences of this definition for the field theory $\textcal{f}$. $\textcal{f}$ takes $M^{d-2}$ to an object of $\textcal{Hom}_{\mathcal{H}_2}(\mathcal{A}(M^{d-2}), \mathcal{H}_1)$, i.e. a $\mathcal{H}_1$-linear functor $\textcal{f}\:(M^{d-2}) : \mathcal{A}(M^{d-2}) \rightarrow \mathcal{H}_1$. Recall also that $\mathcal{A}(M^{d-2})$ is a 2-Hilbert space, and is therefore non-canonically equivalent to $\mathcal{H}_1^n$ \cite{2008arXiv0812.4969B}, the $n$th Cartesian product of the category of Hilbert spaces. As the functor preserves the $\mathcal{H}_1$-module structure, it is determined by its value on the $n$ copies of $\mathbb{C}$ generating $\mathcal{H}_1^n$ as a category module over $\mathcal{H}_1$. Let us write these generators $\mathbb{C}_i$, $i = 1,...,n$. Writing $H_i = \textcal{f}\:(\mathbb{C}_i)$, we get a collection of Hilbert spaces. The anomalous theory is therefore associated to a collection $\{H_i\}$ of Hilbert spaces that depends on a choice of equivalence $\mathcal{A}(M^{d-2}) \sim \mathcal{H}_1^n$. Let us stress that this equivalence can in general \emph{not} be chosen canonically. We will see shortly the consequences of this fact. Let us also mention that the vectors in the Hilbert spaces $H_i$ \emph{cannot} be pictured directly as states of the anomalous theory, because the equivalence used to picture $\mathcal{A}(M^{d-2})$ as $\mathcal{H}_1^n$ discards some information. In particular, the fact that the $H_i$ are all finite dimensional does not mean that we are restricting ourselves to anomalous field theories with finite dimensional state spaces. We discuss this point in more detail below.

Let us move up in dimension and consider a bordism $M^{d-1,1}$. $\textcal{f}\:(M^{d-1,1})$ is a morphism of the category $\textcal{Hom}_{\mathcal{H}_2}(\mathcal{A}(\partial_- M^{d-1,1}), \mathcal{H}_1)$ between the objects $\textcal{f}\:(\partial_- M^{d-1,1})$ and $\textcal{f}\:(\partial_+ M^{d-1,1}) \circ \mathcal{A}(M^{d-1,1})$. We will see in Sections \ref{SecHamAnInvCase} and \ref{SecHamAnGenCase} that this fact implies that the state space of the anomalous theory is a gerbe.

In the case of a closed $d-1$-dimensional manifold $M^{d-1}$, we have $\mathcal{A}(\emptyset) = \mathcal{H}_1$, $\textcal{f}\:(\emptyset)$ is the identity functor and the discussion above simplifies. We see that $\textcal{f}\:(M^{d-1})$ is now a natural transformation between the identity functor and $\mathcal{A}(M^{d-1})$. But the functor $\mathcal{A}(M^{d-1})$ can be pictured as a Hilbert space and the natural transformation is simply a homomorphism $\textcal{f}\:(M^{d-1}): \mathcal{A}(M^{d-1}) \rightarrow \mathbb{C}$. We see therefore that the field theory $\textcal{f}$ does not yield a complex number on closed $d-1$-dimensional manifolds. It does so only after one specifies an element of the Hilbert space $\mathcal{A}(M^{d-1})$. More precisely, the partition function is a vector in $\mathcal{A}(M^{d-1})^\dagger = \mathcal{A}(M^{d-1})$. The simplest case occurs when $\mathcal{A}$ is an invertible field theory. $\mathcal{A}(M^{d-1})$ is a Hermitian line and the ambiguity in the identification of $\mathcal{A}(M^{d-1})$ with $\mathbb{C}$ translates into a phase ambiguity in the definition of the partition function as a complex number. This is the simplest incarnation of an anomaly, occurring for instance in chiral fermionic theories. When we consider families of manifolds, we obtain a Hermitian line bundle over the parameter space, of which the partition function is a section. In general, $\mathcal{A}(M^{d-1})$ can be an arbitrary $n$-dimensional Hilbert space and the partition function can be (non-canonically) pictured as an $n$-component vector. This is the situation that arises for 2-dimensional chiral conformal field theories, or for the six-dimensional (2,0) theories. In these cases, the components of the partition function are traditionally called "conformal blocks".

\paragraph{Anomalous field theories with infinite dimensional state space} Our aim is ultimately to describe physically relevant quantum field theories, whose state spaces are infinite dimensional. Yet all the Hilbert spaces involved in the formalism above are finite dimensional. How can we then treat anomalous theories with infinite dimensional state space?

As we hinted above, it is naive to think of the finite-dimensional Hilbert spaces $H_i$ as being directly related to the states of the anomalous field theory. Indeed, to extract them, we had to pick an equivalence of categories $\mathcal{A}(M^{d-2}) \sim \mathcal{H}_1^n$. But an equivalence is not an equality: the simple objects $V_i$ of $\mathcal{A}(M^{d-2})$ are mapped through the equivalence to the one-dimensional Hilbert spaces $\mathbb{C}_i$, but they may very well be themselves infinite dimensional. The Hilbert spaces $H_i$ should be more appropriately thought of as multiplicity spaces, or Chan-Patton factors in physical parlance, so that the full state space of the theory reads $\bigoplus_i V_i \otimes H_i$.

In the example to be discussed later in Section \ref{SecExamAnomTh}, where $\mathcal{F}$ is a 2-dimensional rational chiral conformal field theory, $\mathcal{A}(M^{d-2})$ is the representation category of a rational vertex algebra. While such a category is equivalent to $\mathcal{H}_1^n$ for some $n$, the simple objects are infinite dimensional Hilbert spaces. 

This shows that the formalism above has no trouble accommodating anomalous theories with infinite-dimensional state spaces. We should however mention that it is not completely universal. For instance in the example of the previous paragraph, it is crucial that the chiral conformal field theory is rational. Non-rational vertex algebras can have representation categories generated by an infinite number of simple objects, and would require a genuinely infinite dimensional formalism to be properly accounted for.

\subsection{The anomaly of the partition function} 

\label{SecAnPartFunc}

We recognized in the previous section some facts familiar from the physical picture of anomalies, such as the fact that the partition function of $\textcal{f}$, for $\mathcal{A}$ invertible, has a phase ambiguity. However, in physical contexts, anomalies are most often pictured as the breaking of some symmetry of the classical field theory in the quantum field theory. We will show here that anomalous field theories can indeed fail to be invariant under the group of automorphisms of the $\mathsf{F}$-structure of the underlying manifold, which we see as the group of potential symmetries of the theory.

\paragraph{The symmetry group} Recall from Appendix \ref{ApGeomBord2Cat} that given an $(\mathsf{E},\mathsf{F})$-manifold $M^{d-2}$, we write $[M^{d-2}]$ for the underlying $\mathsf{F}$-manifold, and ${\rm Aut}_\mathsf{F}([M^{d-2}])$ for the automorphism group of the $\mathsf{F}$-structure on $[M^{d-2}]$. We write $\phi M^{d-2}$ for the isomorphic $(\mathsf{E},\mathsf{F})$-manifold obtained from $M^{d-2}$ by the action of $\phi \in {\rm Aut}_\mathsf{F}([M^{d-2}])$. The bordism category $\mathcal{B}^{d,2}_{\mathsf{E},\mathsf{F}}$ contains limit 1-morphisms that can be pictured as infinitesimal cylinders  $\lim_{\epsilon \rightarrow 0} M^{d-2} \times (-\epsilon, \epsilon)$. The ingoing boundary $-M^{d-2} \times \{-\epsilon\}$ is identified with $-M^{d-2}$ through the identity map and the outgoing boundary $M^{d-2} \times \{\epsilon\}$ is identified with $\phi M^{d-2}$ through $\phi$. Limit 1-morphism provide a realization of the action of ${\rm Aut}_\mathsf{F}([M^{d-2}])$ on the collection of $(\mathsf{E},\mathsf{F})$-manifolds $N$ such that $[N] = [M^{d-2}]$. Analogous statements hold for 1-morphisms. We have an automorphism group ${\rm Aut}_\mathsf{F}([M^{d-1,1}])$ of $\mathsf{F}$-structures on $[M^{d-1,1}]$. We have limit 2-morphisms realizing the action of ${\rm Aut}_\mathsf{F}([M^{d-1,1}])$ on the collection of $(\mathsf{E},\mathsf{F})$-manifolds with the same underlying $\mathsf{F}$-structure as $M^{d-1,1}$. We will abuse the notation and write $\phi$ as well for the limit morphism associated to an automorphism $\phi$.

${\rm Aut}_\mathsf{F}([M])$ should be pictured as a (potential) symmetry of the field theories defined from the data $\mathsf{F}$ on the manifold $M$. We will see that while the anomaly field theory $\mathcal{A}$ is invariant under this symmetry, the anomalous field theory $\textcal{f}$ is not necessarily invariant. We will also show that the lack of invariance of the anomalous theory can be characterized by group cohomology classes of ${\rm Aut}_\mathsf{F}([M])$, recovering results from the physics literature.

\paragraph{The partition function anomaly} The anomalous field theory is a 2-natural transformation $\textcal{f}: \mathcal{A}|_{d-1} \rightarrow \bm{1}|_{d-1}$, where $\mathcal{A}, \bm{1} : \mathcal{B}^{d,2}_{\mathsf{E},\mathsf{F}} \rightarrow \mathcal{H}_2$ are the 2-functors corresponding to the anomaly field theory and the trivial theory, respectively. As before, $|_{d-1}$ denotes the truncation of these functors to manifolds and bordisms of dimension $d-1$ or less. Let us start by considering a closed $d-1$-dimensional manifold $M^{d-1}$. We assume that the anomaly field theory is a field theory defined on $\mathsf{F}$-manifolds, whose field theory functor $\mathcal{A}$ has been pulled-back from $\mathcal{B}^{d,2}_{\mathsf{F}}$ to $\mathcal{B}^{d,2}_{\mathsf{E},\mathsf{F}}$ (see Appendix \ref{ApGeomBord2Cat}). Then $\mathcal{A}(M^{d-1}) = \mathcal{A}(\phi M^{d-1})$, and the compatibility with the gluing of limit morphisms ensures that $\mathcal{A}(\phi)$ form a unitary representation of ${\rm Aut}_\mathsf{F}([M^{d-1}])$ on the state space $\mathcal{A}(M^{d-1})$. A similar reasoning applies to the trivial field theory $\bm{1}$. $\bm{1}(\phi)$ is the identity homomorphism $\mathbb{C} \rightarrow \mathbb{C}$, so $\bm{1}$ provides the trivial representation of ${\rm Aut}_\mathsf{F}([M^{d-1}])$ on $\mathbb{C}$.

Recall that the anomalous field theory provides a homomorphism $\textcal{f}\:(M^{d-1}): \mathcal{A}(M^{d-1}) \rightarrow \mathbb{C}$. The fact that $\textcal{f}$ is a natural transformation requires it to intertwine the representations defined by $\mathcal{A}$ and $\bm{1}$. Concretely, we deduce that 
\be
\label{EqIntertwPartFunc}
\textcal{f}\:(M^{d-1}) = \textcal{f}\:(\phi M^{d-1}) \circ \mathcal{A}(\phi)
\ee
Let us write ${\rm Aut}_{\mathsf{E},\mathsf{F}}(M^{d-1})$ for the group of automorphisms of the $(\mathsf{E},\mathsf{F})$-manifold $M^{d-1}$ in the category $\mathcal{M}_{\mathsf{E},\mathsf{F}}$. We have:
\begin{proposition}
\label{PropPartFuncInvAut}
The vector of partition functions $\textcal{f}\:(M^{d-1})$ vanishes outside the space of invariants of the action of ${\rm Aut}_{\mathsf{E},\mathsf{F}}(M^{d-1})$ on $\mathcal{A}(M^{d-1})$.
\end{proposition}
\begin{proof}
If $\phi \in {\rm Aut}_{\mathsf{E},\mathsf{F}}(M^{d-1})$, then $\phi M^{d-1} = M^{d-1}$. \eqref{EqIntertwPartFunc} then shows that $\textcal{f}\:(M^{d-1})$ is an intertwiner between the representation of ${\rm Aut}_{\mathsf{E},\mathsf{F}}(M^{d-1})$ determined by $\mathcal{A}$ and the trivial representation (determined by $\bm{1}$). Such an intertwiner can be non-vanishing only on the space of invariants.
\end{proof}
We see here the importance of working with manifolds in the category $\mathcal{M}_{\mathsf{E},\mathsf{F}}$ rather than in $\mathcal{M}_{\mathsf{F}}$. The latter case corresponds to setting $\mathsf{E} = \mathsf{F}$, so ${\rm Aut}_{\mathsf{E},\mathsf{F}}(M^{d-1}) = {\rm Aut}_\mathsf{F}([M^{d-1}])$ and the anomalous theory has to be invariant under the symmetry group ${\rm Aut}_\mathsf{F}([M^{d-1}])$ (at the expense of the vanishing of part or all of its partition functions). 

The kind of anomaly leading to a vanishing of the partition function through Proposition \ref{PropPartFuncInvAut} appeared in the physics literature. For instance in \cite{Witten:1999vg}, it was shown that the partition function of a self-dual field vanishes unless certain torsion background fluxes are turned on. The torsion fluxes required to have a non-vanishing partition function are precisely those that make trivial the representation of gauge transformations associated to certain torsion classes on the Hermitian line in which the partition function of the self-dual field takes value. A detailed discussion can be found in Section 3.6 of \cite{Monniera}.

From now on, we assume that the structure $\mathsf{E}$ has been chosen so that ${\rm Aut}_{\mathsf{E},\mathsf{F}}(M^{d-1})$ is the trivial group for all manifolds $M^{d-1}$. A choice for $\mathsf{E}$ satisfying this condition is described in Appendix \ref{ApGeomBord2Cat}. 
In this case, \eqref{EqIntertwPartFunc} always admits non-vanishing solutions. We say that $\textcal{f}$ has an anomalous symmetry if its partition function fails to be invariant under the symmetry group ${\rm Aut}_\mathsf{F}([M^{d-1}])$. The following statement holds almost tautologically:
\begin{proposition}
The theory $\textcal{f}$ has an anomalous symmetry unless $\textcal{f}\:(M^{d-1})$ takes value in the invariants of the action of ${\rm Aut}_\mathsf{F}([M^{d-1}])$ for all $[M^{d-1}]$. If the representation determined by $\mathcal{A}$ is irreducible and non-trivial, $\textcal{f}\:(M^{d-1})$ is invariant under the symmetry ${\rm Aut}_\mathsf{F}([M^{d-1}])$ if and only if it vanishes. 
\end{proposition}


In particular, when $\mathcal{A}$ is invertible, $\mathcal{A}(M^{d-1})$ is 1-dimensional and the failure of invariance of the partition function of $\textcal{f}$ is by a phase. This phase is a character of ${\rm Aut}_\mathsf{F}([M^{d-1}])$, or equivalently a group 1-cocycle on ${\rm Aut}_\mathsf{F}([M^{d-1}])$ valued in $\mathbb{T}$. The anomaly is absent if the associated group cohomology class is trivial, which in degree 1 actually requires that the cocycle itself is trivial. When ${\rm Aut}_\mathsf{F}([M^{d-1}])$ happens to be a Lie group, the corresponding Lie algebra cocycle condition is well-known in the physics literature and goes under the name of the Wess-Zumino consistency condition, see for instance Section 22.6 of \cite{weinberg1996quantum}. The group cocycle itself was first described in \cite{Witten:1983tw}. We have therefore recovered the familiar physical picture of the anomaly. 

When the anomaly theory $\mathcal{A}$ is not invertible, the vector of partition functions transforms in the unitary representation of ${\rm Aut}_\mathsf{F}([M^{d-1}])$ on $\mathcal{A}(M^{d-1})$. Such a representation can be pictured as a non-abelian group 1-cocycle on ${\rm Aut}_\mathsf{F}([M^{d-1}])$ valued in $U(n)$. This is for instance familiar in the case of 2-dimensional chiral conformal field theories. $M^{d-1}$ is then a 2-dimensional surface endowed with a conformal structure. ${\rm Aut}_\mathsf{F}([M^{d-1}])$ includes the modular group, which is the group of diffeomorphisms preserving a given conformal structure. The conformal blocks of the theory are not invariant, but transform in a unitary representation of the modular group. On the torus, the modular group is isomorphic to $SL(2,\mathbb{Z})$ and the representation is determined by the $S$ and $T$ matrices corresponding to the generators of $SL(2,\mathbb{Z})$, see for instance Chapter 10 of \cite{CFT1997}.

\subsection{The Hamiltonian anomaly in the invertible case} 

\label{SecHamAnInvCase}

We now investigate what happens on a $d-2$-dimensional manifold $M^{d-2}$. When $\mathcal{A}$ is invertible, we expect to recover the Hamiltonian anomaly, i.e. the fact that the symmetry group is represented on the state space of the theory by a projective representation, characterized by a group cohomology class of degree 2 valued in $\mathbb{T}$ \cite{Faddeev:1984jp, Faddeev:1985iz, Mickelsson1987}. This result was obtained recently for topological field theories using similar ideas in \cite{2014arXiv1409.5723F}. To understand how Hamiltonian anomalies arise, let us assume again that the anomaly field theory functor is a pull-back from $\mathcal{B}^{d,2}_\mathsf{F}$ so that that $\mathcal{A}(M^{d-2}) = \mathcal{A}(\phi M^{d-2})$ for $\phi \in {\rm Aut}_\mathsf{F}([M^{d-2}])$. The anomaly field theory is therefore invariant under the symmetry group ${\rm Aut}_\mathsf{F}([M^{d-2}])$. In complete analogy to what happened in the previous section, $\mathcal{A}$ defines a "2-representation" of ${\rm Aut}_\mathsf{F}([M^{d-2}])$ on $\mathcal{A}(M^{d-2})$. By this, we mean that for each automorphism $\phi \in {\rm Aut}_\mathsf{F}([M^{d-2}])$, we have a functor $\mathcal{A}(\phi): \mathcal{A}(M^{d-2}) \rightarrow \mathcal{A}(M^{d-2})$, and the composition of these functors reproduces the group law of ${\rm Aut}_\mathsf{F}([M^{d-2}])$: $\mathcal{A}(\phi_1) \circ \mathcal{A}(\phi_1) = \mathcal{A}(\phi_1\phi_2)$. $\bm{1}$ defines a trivial 2-representation of ${\rm Aut}_\mathsf{F}([M^{d-2}])$ on $\mathcal{H}_1$. The fact that $\textcal{f}$ is a 2-natural transformation implies that we have a natural transformation 
\be
\label{EqNatTransFPhi}
\textcal{f}\:(\phi): \textcal{f}\:(M^{d-2}) \rightarrow  \textcal{f}\:(\phi M^{d-2}) \circ \mathcal{A}(\phi) \;,
\ee
where now $\circ$ is the composition of functors. We have an analogue of Proposition \ref{PropPartFuncInvAut}:
\begin{proposition}
\label{PropRestStSpInv}
The state space $\textcal{f}\:(M^{d-2})$ belongs to the subcategory of $\mathcal{A}(M^{d-2})$ invariant under the action of ${\rm Aut}_{\mathsf{E},\mathsf{F}}(M^{d-2})$.
\end{proposition} 
In what follows, we assume that the $\mathsf{E}$-structure was chosen so that ${\rm Aut}_{\mathsf{E},\mathsf{F}}(M^{d-2})$ is trivial, so that Proposition \ref{PropRestStSpInv} provides no constraint. $\mathsf{E}$-structures with this property are described in Appendix \ref{ApGeomBord2Cat}.

The above is valid both for invertible and non-invertible anomaly field theories. Let us now focus on the case when $\mathcal{A}$ is an invertible field theory. Then $\mathcal{A}(M^{d-2})$ is a 2-Hermitian line, a category non-canonically equivalent to $\mathcal{H}_1$. Pick an equivalence $\chi$. $\mathcal{A}(\phi)$ can now be pictured as an invertible $\mathcal{H}_1$-linear functor $\mathcal{H}_1 \rightarrow \mathcal{H}_1$. Such a functor is the tensor product with a Hermitian line $L_{\chi,\phi}$. As we have a 2-representation, we have a canonical isomorphism 
\be
\label{EqIsomTensLineBunGerbe}
L_{\chi,\phi_1} \otimes L_{\chi,\phi_2} \otimes L_{\chi,\phi_2^{-1} \phi_1^{-1}} \simeq \mathbb{C} \;.
\ee 
Pick in addition for each $L_{\chi,\phi}$ a non-canonical isomorphism $L_{\chi,\phi} \simeq \mathbb{C}$, and for notational convenience, include this extra data in the symbol $\chi$. The isomorphism \eqref{EqIsomTensLineBunGerbe} is then a unitary transformation $\mathbb{C} \rightarrow \mathbb{C}$, i.e. an element $\alpha_{\chi,\phi_1, \phi_2}$ of $\mathbb{T}$. Standard arguments show that $\alpha$ is a 2-cocycle for the group ${\rm Aut}_\mathsf{F}([M^{d-2}])$. This is the 2-cocycle described by Faddeev in \cite{Faddeev:1984jp}, and whose infinitesimal version was described by Mickelsson in \cite{Mickelsson1985}. The cocycle itself is dependent on the choices of equivalence and isomorphisms $\chi$, but its group cohomology class is not. These two claims are consequences of the general result proven in the next section for the case of a not necessarily invertible anomaly field theory.

We use the chosen equivalence of $\mathcal{A}(M^{d-2})$ with $\mathcal{H}_1$ to see $\textcal{f}\:(M^{d-2})$ as an $\mathcal{H}_1$-linear functor from $\mathcal{H}_1$ to itself, hence as a Hilbert space $H_\chi(M^{d-2})$. \eqref{EqNatTransFPhi} can be rewritten as an isomorphism
\be
f(\phi)|_\chi : H_\chi(M^{d-2}) \rightarrow H_\chi(\phi M^{d-2}) \otimes L_{\chi,\phi} \;.
\ee
We now see that given two group elements $\phi_1, \phi_2 \in {\rm Aut}_\mathsf{F}([M^{d-2}])$,
\be
f(\phi_2^{-1} \phi_1^{-1})|_\chi \circ f(\phi_2)|_\chi \circ f(\phi_1)|_\chi : H_\chi(M^{d-2}) \rightarrow H_\chi(M^{d-2})
\ee
is given by the multiplication by the cocycle $\alpha_{\chi, \phi_1, \phi_2}$. We recovered the fact that for invertible anomalies, the representation of ${\rm Aut}_\mathsf{F}([M^{d-2}])$ on the state space of the anomalous theory is only a projective one, characterized by the 2-cocycle $\alpha$.

\subsection{The Hamiltonian anomaly in the general case} 

\label{SecHamAnGenCase}

It is interesting to consider what happens when the anomaly theory is not invertible. As far as we are aware, this situation has not been described in the physics literature yet and this is the case relevant for six-dimensional (2,0) theories. 

\paragraph{Unpacking the definitions} When $\mathcal{A}$ is not invertible, we still have natural transformations \eqref{EqNatTransFPhi}. But now $\mathcal{A}(M^{d-2})$ is non-canonically equivalent to $\mathcal{H}_1^n$, on which the functors $\mathcal{A}(\phi)$ provide a 2-representation of ${\rm Aut}_\mathsf{F}([M^{d-2}])$. Let us pick again an equivalence $\chi$. A generic $\mathcal{H}_1$-linear functor from $\mathcal{H}_1^n$ to itself can be represented as a matrix of Hilbert spaces \cite{2008arXiv0812.4969B}. The invertibility of the functors $\mathcal{A}(\phi)$ implies two facts. First, their matrix elements $L_{\chi,\phi}^{ij}$, $1 \leq i,j \leq n$ can only be either Hermitian lines or the zero-dimensional Hilbert space. In the latter case, we say that the matrix element is "vanishing". Moreover, as $\mathcal{H}_1^n$ is a semiring and not a ring, invertibility also requires the matrix to be a permutation matrix, i.e. that there is a single non-vanishing entry on each line and column. Let us write $L_{\chi,\phi}$ for the matrix with matrix elements $L_{\chi,\phi}^{ij}$ and $\ell_{\chi,\phi}$ for the associated permutation matrix, i.e. the matrix obtained from $L_{\chi,\phi}$ by replacing Hermitian lines by $1$ and the zero vector space by $0$. Let $\boxtimes$ be the combination of the tensor product and the matrix multiplication, i.e.:
\be
(L_1 \boxtimes L_2)^{ik} = \bigoplus_{j = 1}^n L_1^{ij} \otimes L_2^{jk} \;.
\ee
The fact that we have a 2-representation means that there are canonical isomorphisms 
\be
\label{Eq2RepNonInvCase}
L_{\chi,\phi_1} \boxtimes L_{\chi,\phi_2} \boxtimes L_{\chi,\phi_2^{-1} \phi_1^{-1}} \simeq \mathds{1} \;,
\ee
where $\mathds{1}$ is the matrix of Hilbert space that has copies of $\mathbb{C}$ on the diagonal and vanishing matrix elements off the diagonal. 

Again, let us pick for each $L_{\chi,\phi}^{ij}$ that is different from the zero Hilbert space a non-canonical isomorphism  $L_{\chi,\phi}^{ij} \simeq \mathbb{C}$, which we include in the data $\chi$. \eqref{Eq2RepNonInvCase} then provides a $\mathbb{C}$-valued unitary matrix $\alpha^{ik}_{\chi, \phi_1,\phi_2}$, which is the product of a permutation matrix with a diagonal matrix with entries in $\mathbb{T}$, i.e. we obtain an element $\alpha_{\chi, \phi_1,\phi_2} \in S_n \ltimes \mathbb{T}^n$. (The semi-direct product is with respect to the permutation action of the symmetric group on $\mathbb{T}^n$.)

\paragraph{The non-abelian 2-cocycle} Writing $\lambda_{\chi,\phi} = {\rm Ad}(\ell_{\chi,\phi})$ for the adjoint action of $\ell_{\chi,\phi}$ on $S_n \ltimes \mathbb{T}^n$, it is not difficult to check that we have the relation (dropping the mention of the data $\chi$ to lighten the notation)
\be
\lambda_{\phi_1} \lambda_{\phi_2} \lambda_{\phi_2^{-1}\phi_1^{-1}} = {\rm Ad}(\alpha_{\phi_1,\phi_2}) \;.
\ee
Moreover, there are two different ways to use the isomorphism \eqref{Eq2RepNonInvCase} to identify $L_{\chi,\phi_1\phi_2\phi_3}$ with $L_{\chi,\phi_1} \boxtimes L_{\chi,\phi_2} \boxtimes L_{\chi,\phi_3}$, which leads to the relation
\be
\lambda_{\phi_1}(\alpha_{\phi_2, \phi_3}) \alpha_{\phi_1, \phi_2\phi_3} = \alpha_{\phi_1,\phi_2} \alpha_{\phi_1\phi_2, \phi_3} \;.
\ee
Comparing for instance with (5.1.10) of \cite{2006math.....11317B}, we see that the pair $(\lambda_\chi, \alpha_\chi)$ satisfies the same relations as the cocycle associated to a non-abelian gerbe.

(The discussion in \cite{2006math.....11317B} pertains to bundle gerbes over a topological space. In comparing with \cite{2006math.....11317B}, we must keep in mind that in our case, the gerbe effectively lives on the classifying space of the symmetry group $G = {\rm Aut}_\mathsf{F}([M^{d-2}])$. Recall that a classifying space $BG$ is the quotient of a contractible space $EG$ by a free action of $G$. Practically, this means that we can cover $EG$ with a unique chart that is acted on by $G$. We therefore identify pairs of cover indices in \cite{2006math.....11317B} with elements of $G$. For instance, if we identify $(ij)$ with $\phi_1$, $(jk)$ with $\phi_2$ and $(kl)$ with $\phi_3$, an object $X_{ijl}$ living on a triple intersection corresponds in our framework to an object $X_{\phi_1, \phi_2\phi_3}$.)

To show that the cohomology class of this cocycle is independent of the extra choices that we have collectively written $\chi$, we must study the dependence of $(\lambda_\chi, \alpha_\chi)$ on the latter. We made essentially two types of choices. The first one was the choice of equivalence between $\mathcal{A}(M^{d-2})$ and $\mathcal{H}_1^n$. The second one was the choice of isomorphisms $L^{ij}_{\chi,\phi} \simeq \mathbb{C}$. If we choose two equivalences between $\mathcal{A}(M^{d-2})$ and $\mathcal{H}_1^n$, they will differ by a permutation $r \in S_n$ of the generators of $\mathcal{H}_1^n$. We call $\chi$ and $\chi'$ the data encoding the two choices of equivalence. Writing $\rho = {\rm Ad}(r)$, we see that the cocycles are related by
\be
\left(\lambda_{\chi'}, \alpha_{\chi'}\right) = \left(\rho \lambda_{\chi} \rho^{-1}, \rho(\alpha_{\chi})\right) \;.
\ee
Comparing with (5.2.9) of \cite{2006math.....11317B} we see that $(\lambda_{\chi'}, \alpha_{\chi'})$ is cohomologous to $(\lambda_{\chi}, \alpha_{\chi})$. Suppose now that we change the isomorphisms $L_{\chi,\phi}^{ij} \simeq \mathbb{C}$. We can encode the changes of isomorphisms by elements $\theta_\phi \in \mathbb{T}^n$. We call again $\chi'$ the new data. The new cocycle reads
\be
\left(\{\lambda_{\chi',\phi}\}, \{\alpha_{\chi', \phi_1,\phi_2}\}\right) = \left(\{\lambda_{\chi,\phi} \}, \{\theta_{\phi_1} \lambda_{\chi,\phi_1}(\theta_{\phi_2}) \alpha_{\chi,\phi_1,\phi_2} \theta^{-1}_{\phi_1 \phi_2} \} \right) \;,
\ee
where we picture the $\theta$'s as diagonal matrices and use the matrix product, i.e $(\theta \alpha)^{ij} := \theta^i \alpha^{ij}$ and $(\alpha \theta)^{ij} = \theta^j \alpha^{ij}$.
Comparing with (5.2.9) of \cite{2006math.....11317B} we see again that $(\lambda_{\chi'}, \alpha_{\chi'})$ is cohomologous to $(\lambda_{\chi}, \alpha_{\chi})$.

The relevant non-abelian cohomology theory is the \v{C}ech hypercohomology 
\be
H:= \check{H}\left(B{\rm Aut}_{\mathsf{F}}([M^{d-2}]); S_n \ltimes \mathbb{T}^n \rightarrow S_n \right)
\ee
whose definition can be found in \cite{2006math.....11317B}. We have therefore proved:
\begin{proposition}
\label{ThHamAnCohomClass}
A $d-1$-dimensional anomalous field theory admitting a symmetry group ${\rm Aut}_{\mathsf{F}}([M^{d-2}])$ on a $d-2$ dimensional manifold $M^{d-2}$ determines a cohomology class in $H$.
\end{proposition}
Remark that with an invertible anomaly field theory, when $n = 1$, $\lambda$ does not contain any information and $\alpha$ is an ordinary $\mathbb{T}$-valued cocycle. In this case, the cohomology reduces to the ordinary group cohomology of ${\rm Aut}_{\mathsf{F}}([M^{d-2}])$ valued in $\mathbb{T}$. The fact that Hamiltonian anomalies can be characterized in this way has been known for some time \cite{Faddeev:1984jp}.

\paragraph{The twisted representation of the symmetry group on the vector of Hilbert spaces} Let us study the consequence of the above discussion for the action of the symmetry group on the state space of the theory. With our choice of non-canonical equivalence of $\mathcal{A}(M^{d-2})$ with $\mathcal{H}_1^n$, we can see $\textcal{f}\:(M^{d-2})$ as an additive functor from $\mathcal{H}_1^n$ to $\mathcal{H}_1$, hence as a vector of Hilbert spaces $H^i_\chi(M^{d-2})$. \eqref{EqNatTransFPhi} can be rewritten as an isomorphism
\be
f(\phi)|_\chi : H_\chi(M^{d-2}) \rightarrow H_\chi(\phi M^{d-2}) \boxtimes L_{\chi,\phi}\;,
\ee
where we wrote $H_\chi(M^{d-2})$ for the vector of Hilbert spaces. Remark that there is again only one term contributing non-trivially in the direct sum implicitly present on the right-hand side.
We now see that given two group elements $\phi_1, \phi_2 \in {\rm Aut}_\mathsf{F}([M^{d-2}])$,
\be
f(\phi_2^{-1} \phi_1^{-1})|_\chi \circ f(\phi_2)|_\chi \circ f(\phi_1)|_\chi : H_\chi(M^{d-2}) \rightarrow H_\chi(M^{d-2})
\ee
is given by the multiplication by the cocycle $\alpha_{\chi, \phi_1, \phi_2}$, where the multiplication also involves a matrix multiplication on the indices of the Hilbert spaces. We therefore discover that when the anomaly field theory is non-invertible, in general we do not get a representation of the symmetries on a single Hilbert space. We only get a projective representation on the vector $H_\chi(M^{d-2})$ of Hilbert space, in the sense described above. Note that this is in total analogy to what happens for the partition functions (or conformal blocks): the symmetry group does not have an action on a single partition function, but may permute them. The only difference is that while arbitrary unitary transformations of the space of partition functions can occur, the semiring structure of $\mathcal{H}_1$ reduces such transformations to elements of the group $S_n \ltimes \mathbb{T}^n$.

Let us also mention that when the cohomology class of Proposition \ref{ThHamAnCohomClass} is trivial, the proof of the proposition shows that the choices $\chi$ can be made in such a way that $\lambda_{\chi,\phi} = 1 \in S_n$ and $\alpha_{\chi,\phi_1,\phi_2} = 1 \in S_n \ltimes \mathbb{T}^n$. In this case we see that we get a linear representation of the symmetry group ${\rm Aut}_{\mathsf{F}}([M^{d-2}])$ on each of the Hilbert spaces $H^i_\chi(M^{d-2})$, and the symmetry is not anomalous. We therefore obtain 
\begin{proposition}
\label{ThNonAbCohomClassCharAnom}
If the representation of ${\rm Aut}_{\mathsf{F}}([M^{d-2}])$ on $\mathcal{A}(M^{d-2})$ is irreducible, the symmetry ${\rm Aut}_{\mathsf{F}}([M^{d-2}])$ is anomalous if and only if the cohomology class of Proposition \ref{ThHamAnCohomClass} is non-trivial.
\end{proposition}
The irreducibility requirement is there to eliminate the possibility that $\mathcal{A}(M^{d-2})$ be a reducible 2-representation. In this case, even if $\mathcal{A}(M^{d-2})$ is non-trivial and associated to a non-zero cohomology class, $\textcal{f}\:(M^{d-2})$ could take value in the possibly non-vanishing invariant subcategory of $\mathcal{A}(M^{d-2})$ and be invariant under ${\rm Aut}_{\mathsf{F}}([M^{d-2}])$. Proposition \ref{ThNonAbCohomClassCharAnom} generalizes a corresponding well-known statement in the case where the anomaly field theory is invertible.

Let us finally mention that the picture above should generalize to families of manifolds and bordisms. To discuss these rigorously, we would need a better definition of the cobordism category, presumably along the lines of \cite{Ayal}, in which one obtains a natural topology on the moduli spaces of $\mathsf{F}$-manifolds and bordisms. We expect that the non-abelian gerbe characterized by the cohomology class of Proposition \ref{ThHamAnCohomClass}, which in our setup is defined over the classifying space of ${\rm Aut}_{\mathsf{F}}([M^{d-2}])$, should be promoted to a non-abelian bundle gerbe over the moduli space of $d-2$-dimensional $\mathsf{F}$-manifolds.

\subsection{Example} 

\label{SecExamAnomTh}

As was mentioned above, the prime examples of anomalous field theories with non-invertible anomaly field theories are 2-dimensional rational chiral conformal field theories and the six-dimensional (2,0) superconformal field theories. The anomaly field theories of the latter involve a certain refinement of the Dijkgraaf-Witten theory \cite{Dijkgraaf:1989pz}, which is unfortunately not completely straightforward to construct and which will be studied elsewhere.

The case of a 2-dimensional rational chiral conformal field theory can be treated in complete generality as follows. We do not make explicit many of the concepts and constructions in what follows, see for instance the books \cite{bakalov2001lectures, CFT1997}. Let $\mathcal{C}$ be a modular tensor category. To fix ideas, one can keep in mind the case of chiral Wess-Zumino-Witten theories, in which $\mathcal{C}$ is the category of positive energy representations of the level $k \in \mathbb{N}$ central extension of the loop group of a semi-simple Lie group $G$. Then on the one hand, $\mathcal{C}$ contains the Moore-Seiberg data \cite{Moore:1988qv} required to define a 2-dimensional rational chiral conformal field theory $\mathcal{R}_\mathcal{C}$ \cite{2005PNAS..102.5352H, Fuchs:2002cm, Fuchs:2004xi}. On the other hand, the Reshetikhin-Turaev construction \cite{Reshetikhin1991, turaev1994quantum} provides a 3-dimensional topological field theory $\mathcal{A}_\mathcal{C}$. 

$\mathcal{R}_\mathcal{C}$ is an anomalous field theory in the sense above, whose (generally non-invertible) anomaly field theory is $\mathcal{A}_\mathcal{C}$. Indeed, the chiral conformal field theory $\mathcal{R}_\mathcal{C}$ does not have a well-defined partition function, but a vector of "conformal blocks", which takes value in the state space of the Reshetikhin-Turaev theory $\mathcal{A}_\mathcal{C}$. As we already mentioned, the automorphism group of a surface up to conformal transformations includes the modular group, which is the group of diffeomorphisms preserving a given conformal structure. It is known that in general, the latter acts non-trivially on the space of conformal blocks. 

In codimension 2, $\mathcal{A}_\mathcal{C}(S^1) = \mathcal{C}$, the modular tensor category itself.  If $\mathcal{C}$ has $n$ simple objects, we have an equivalence $\mathcal{C} \sim \mathcal{H}_1^n$. As was discussed above, we can extract from the chiral conformal field theory $\mathcal{R}_\mathcal{C}$ a vector of Hilbert spaces $H_i$, $i = 1,...n$, which should be thought of as the multiplicities with which the simple objects of $\mathcal{C}$ occur in the spectrum of the chiral theory. The group ${\rm Aut}_\mathsf{F}([S^1])$ contains in particular the group of (orientation preserving) diffeomorphisms of the circle ${\rm Diff}^+(S^1)$. There is a gravitational anomaly characterized by the central charge of the chiral theory and whose effect is that ${\rm Diff}^+(S^1)$ is represented only projectively on the state space. As ${\rm Diff}^+(S^1)$ is connected, the degree 2 group cocycle associated with the gravitational anomaly does not involve non-trivial permutation matrices in $S_n \ltimes \mathbb{T}^n$. The non-abelian nature of the anomaly is therefore not manifest in this example. 

It should also be mentioned that in this particular example, $\mathcal{A}_\mathcal{C}$ is more than simply the anomaly field theory of $\mathcal{R}_\mathcal{C}$; the two theories are equivalent, in the sense that there is a prescription allowing the computation of any correlator of $\mathcal{R}_\mathcal{C}$ in $\mathcal{A}_\mathcal{C}$ (see for instance \cite{Fuchs:2002cm, Fuchs:2004xi}). This is of course not the case for generic pairs of anomalous and anomaly field theories.

It could be interesting to work out this example in more detail, but our main aim being the (2,0) theories, we will refrain from doing so. The rest of the paper is devoted to the construction of invertible anomaly field theories that are of physical interest.

\section{Wess-Zumino field theories}

\label{SecWZTh}

We present in this section a class of extended field theories describing the anomalies produced by the "Wess-Zumino terms" of the physics literature.

Our construction generalizes the construction of the classical Dijkgraaf-Witten theory by Freed in \cite{Freed:1994ad} and is strongly inspired by this work. Note that such theories have been constructed using elaborate technology under the name of $\infty$-Chern-Simons theories, see for instance \cite{Fiorenza:2012ec} or \cite{Fiorenza:2013jz}. 

Note also that as we are here describing only anomaly field theories, we do not need $\mathsf{E}$-structures and omit them to simplify the discussion. We  write $\mathsf{F}$-manifolds without square brackets.

\subsection{Definition}

Assume that $\mathsf{F}$ is a structure such that on a manifold $M$, $\mathsf{F}(M)$ includes the data of a $\mathbb{R}/\mathbb{Z}$-valued cocycle $\hat{c}$ of degree $d+1$ (see Appendix \ref{AppDiffCoc}). This typically occurs in two situations
\begin{itemize} 
\item Assume that $\mathsf{F}(M)$ includes the data of a map of $M$ into some classifying space endowed with such a cocycle. We obtain a cocycle on $M$ by pull-back via the classifying map.
\item Assume that $\mathsf{F}(M)$ include differential cocycles modeling gauge fields. We may be able to combine the latter into a degree $d+1$ differential cocycle, whose field strength vanishes for dimensional reason. Such a differential cocycle is equivalent to a degree $d$ $\mathbb{R}/\mathbb{Z}$-valued cocycle.
\end{itemize}
We obtain an extended field theory essentially by integrating $\hat{c}$. More precisely, the Wess-Zumino field theories are defined as follows.

\paragraph{Closed $d$-dimensional manifolds} We can integrate $\hat{c}$ over $M^d$ to obtain an element of $\mathbb{R}/\mathbb{Z}$. Picking a cycle representative $\hat{m}$ of the fundamental class of $M^d$, we define
\be
\label{EqDefWZdDimMan}
\textcal{Int}(\hat{c}, M^d) := \exp 2\pi i \langle \hat{c}, \hat{m} \rangle \in \mathbb{T} \subset \mathbb{C} \;,
\ee
where the angular bracket denotes the pairing between chains and cochains. This expression depends only on the cohomology class $c$ of $\hat{c}$. We define the partition function of the Wess-Zumino field theory associated to $c$ by 
\be
\label{EqDefWZdDimMan2}
\mathcal{W\!\!Z}_{c}(M^d) := \textcal{Int}(\hat{c}, M^d) \;.
\ee

\paragraph{Closed $d-1$-dimensional manifolds} Consider now a closed $d-1$-dimensional $\mathsf{F}$-manifold $M^{d-1}$ and pick a differential cocycle $\check{c} = (\hat{a}, \hat{h}, 0)$ of degree $d+1$ lifting the (necessarily trivial) cohomology class $c$. We cannot integrate $\hat{h}$ over $M^{d-1}$ in the usual sense of the term, but we can do the following \cite{Freed:1994ad}. 

Consider the category $\mathcal{C}_1$ whose objects are cycle representatives of the fundamental homology class of $M^{d-1}$, and whose morphisms between cycles $\hat{m}_1$ and $\hat{m}_2$ are chains $\hat{n}$ on $M^{d-1}$ with boundary $\partial \hat{n} = \hat{m}_2 - \hat{m}_1$. We can construct a functor $\mathcal{F}_1$ to the category $\mathcal{H}_1$ of Hilbert spaces as follows. $\mathcal{F}_1$ takes any object of $\mathcal{C}_1$ to $\mathbb{C}$, and takes any chain $\hat{n}$ to $\exp 2\pi i \langle \hat{n}, \hat{h} \rangle \in \mathbb{T}$, where $\langle \bullet , \bullet \rangle$ is the pairing between chains and cochains. The inverse limit construction, reviewed in Appendix \ref{AppInvLimConstr}, provides us with a Hermitian line $\textcal{Int}(\check{c}, M^{d-1})$, defined as the space of invariant sections of the functor $\mathcal{F}_1$.

This Hermitian line depends a priori on the cocycle representative $\check{c}$. We can eliminate this dependence as follows, with another inverse limit construction. Recall that $\check{\mathcal{Z}}^{p}(M)$ is the category of differential cocycles of degree $p$ on $M$, defined in Appendix \ref{AppDiffCoc}. Consider the functor $\mathcal{G}_1$ that sends objects $\check{c} \in \check{\mathcal{Z}}^{d+1}(M^{d-1})$ to $\textcal{Int}(\check{c}, M^{d-1}) \in \mathcal{H}_1$, and morphisms $\check{c}'$ to $\textcal{Int}(\check{c}', M^{d-1}) \in \mathbb{C}$. By taking the inverse limit of $\mathcal{G}_1$, we obtain a Hermitian line $\mathcal{W\!\!Z}_c(M^{d-1})$ that depends only on the cohomology class $c$. But there is in fact only a single such cohomology class, the trivial one. $\mathcal{W\!\!Z}_c(M^{d-1})$ is therefore a line canonically associated to $M^{d-1}$. One can trivialize it by picking both a cocycle $\check{c}$ and a chain $\hat{m}$.

\paragraph{$d$-dimensional manifolds with boundary} On a $d$-dimensional manifold with boundary $M^{d,1}$, we can define $\textcal{Int}(\check{c}, M^{d,1})$ by pairing as above $\hat{h}$ with any (relative) chain representative $\hat{p}$ of the fundamental homology class $[M^{d,1},\partial M^{d,1}]$. However, as $\hat{h}$ is in general not a relative cocycle, the value of the pairing depends on $\partial\hat{p}$. Under a change of cycle from $\hat{p}$ to $\hat{p}'$ with $\partial \hat{p} = \hat{m}_1$ and $\partial \hat{p}' = \hat{m}_2$, the value of the pairing gets multiplied by $\exp 2\pi i \langle \hat{n}, \hat{h} \rangle$, where $\hat{n}$ is any chain satisfying $\partial \hat{n} = \hat{m_2} - \hat{m}_1$. We see therefore that $\textcal{Int}(\check{c},M^{d,1})$ is not a well-defined complex number. It is rather an invariant section of the functor $\mathcal{F}_1$ associated to the boundary as defined in the previous paragraph. Therefore $\textcal{Int}(\check{c}, M^{d,1}) \in \textcal{Int}(\check{c}|_{\partial M^{d,1}}, \partial M^{d,1})$.

One can now check that the assignment $\check{c} \mapsto \textcal{Int}(\check{c}, M^{d,1})$ is an invariant section of the functor $\mathcal{G}_1$ associated to the boundary, which we define to be $\mathcal{W\!\!Z}_c(M^{d,1})$. It follows that
\be
\mathcal{W\!\!Z}_c(M^{d,1}) \in \mathcal{W\!\!Z}_c(\partial M^{d,1}) \;.
\ee

\paragraph{Closed $d-2$-dimensional manifolds} We pick again a differential cocycle $\check{c} = (\hat{a}, \hat{h}, 0)$ of degree $d+1$, lifting the (trivial) differential cohomology class $c$. The integral of the real cocycle $\hat{h}$ on a closed manifold of dimension $d-2$ is a 2-Hermitian line, constructed as follows. We consider the 2-category $\mathcal{C}_2$ whose objects are cycle representatives of the fundamental class of $M^{d-2}$. A 1-morphism between the objects $\hat{m}_1$ and $\hat{m}_2$ is a degree $d-1$ chain $\hat{n}$ with $\partial \hat{n} = \hat{m}_2 - \hat{m}_1$. A 2-morphism between two such chains $\hat{n}_1$ and $\hat{n}_2$ is a chain $\hat{p}$ with $\partial \hat{p} = \hat{n}_1 - \hat{n}_2$. (Note that the right-hand side is closed, as $\hat{n}_1$ and $\hat{n}_2$ have the same boundary.) 
A 1- or 2-morphism from a cycle/chain to itself is a cycle of one degree higher, and the units are the zero cycles. Composition is simply the addition of cycles and chains, and is strict. We define a 2-functor $\mathcal{F}_2$ from $\mathcal{C}_2$ into the 2-category $\mathcal{H}_2$ of 2-Hilbert spaces. $\mathcal{F}_2$ sends each object to $\mathcal{H}_1$, seen as an object of $\mathcal{H}_2$. It sends each 1-morphism to $\mathbb{C}$, seen as the trivial functor from $\mathcal{H}_1$ to itself. It sends a 2-morphism $\hat{p}$ to $\exp 2\pi i \langle \hat{p}, \hat{h} \rangle \in \mathbb{T}$, seen as a natural transformation between two copies of the trivial functor. 

We now use the inverse limit construction of Appendix \ref{AppInvLimConstr}. Given two objects $\hat{m}_1$ and $\hat{m}_2$, we construct a Hermitian line $L_{\hat{m}_1,\hat{m}_2}$ as the inverse limit of the restriction of $\mathcal{F}_2$ to the category of morphisms from $\hat{m}_1$ to $\hat{m}_2$. $L_{\hat{m}_1,\hat{m}_2}$ is composed of invariant sections $s$ such that given two 1-morphisms $\hat{n}_1$ and $\hat{n}_2$ as above, $s(\hat{n}_2) = s(\hat{n}_1) \exp 2\pi i \langle \hat{p}, \hat{h} \rangle$. 

Consider invariant sections of the functor $\mathcal{F}_2$ defined as follows. To each object $\hat{m}$ of $\mathcal{C}_2$, we associate a Hermitian line  $L_{\hat{m}}$ such that for any two objects $\hat{m}_1$ and $\hat{m}_2$, $L_{\hat{m}_2} = L_{\hat{m}_1,\hat{m}_2} \otimes L_{\hat{m}_1}$. We define the integral $\mathcal{I}$ of $\hat{h}$ over $M^{d-2}$ to be the inverse limit of $\mathcal{F}_2$. To understand what kind of object this is, remark that the set of invariant sections provides us with something that looks like a Hermitian line, but that cannot be canonically identified with any Hermitian line. By tensoring $\mathcal{I}$ with elements of $\mathcal{H}_1$, we obtain a category $\textcal{Int}(\check{c}, M^{d-2})$ that is non-canonically equivalent to $\mathcal{H}_1$. This is what we defined as a 2-Hermitian line in Appendix \ref{App2Vect}. 

An inverse limit procedure for the functor $\mathcal{G}_2 = \textcal{Int}(\bullet, M^{d-2}): \mathcal{Z}^{d+1}(M^{d-2}) \rightarrow \mathcal{H}_2$ yields a 2-Hermitian line $\mathcal{W\!\!Z}_c(M^{d-2})$ in a way completely analogous to the case of closed manifolds of dimension $d-1$. As the differential cohomology class $c$ is necessarily trivial, $\mathcal{W\!\!Z}_c(M^{d-2})$ is a 2-Hermitian line canonically associated to $M^{d-2}$.

\paragraph{$d-1$-dimensional manifolds with boundary} We can repeat the discussion for a closed $d-1$-dimensional manifold, using a chain representative $\hat{r}$ for the fundamental homology class $[M^{d-1,1},\partial M^{d-1,1}]$. We obtain in this way a Hermitian line $L_{\hat{r}}$. However, $L_{\hat{r}}$ depends on $\partial \hat{r}$. In fact, if $\hat{r}_1$ and $\hat{r}_2$ are two such chains, with $\partial \hat{r}_1 = \hat{m}_1$ and $\partial \hat{r}_2 = \hat{m}_2$, then $L_{\hat{r}_2} = L_{\hat{m}_1,\hat{m}_2} \otimes L_{\hat{r}_1}$. This shows the integral $\textcal{Int}(\check{c}, M^{d-1,1})$ is an object of the 2-Hermitian line $\textcal{Int}(\check{c}|_{\partial M^{d-1,1}}, \partial M^{d-1,1})$.

Again one checks that the map $\check{c} \mapsto \textcal{Int}(\check{c}, M^{d-1,1})$ is an invariant section of the functor $\mathcal{G}_2$ above, which we define to be $\mathcal{W\!\!Z}_c(M^{d-1,1})$. We then have $\mathcal{W\!\!Z}_c(M^{d-1,1}) \in \mathcal{W\!\!Z}_c(\partial M^{d-1,1})$.

\paragraph{$d$-dimensional manifolds with corners} This case can be treated analogously to the case of $d$-dimensional manifolds with boundary, as we are integrating on singular cycles anyway. 

Assume that $M^{d,2}$ is a manifold with corners associated to a 2-morphism in $\mathcal{B}^{d,2}_{\mathsf{F}}$. This means that we have two (possibly empty) codimension 2 corners $M_1$ and $M_2$ which are smooth closed $d-2$-dimensional manifolds, two manifolds $N_1$ and $N_2$ that both have $M_2 \sqcup -M_1$ as boundary, and a bordism $P$ between them. Let $\hat{m}_1$, $\hat{m_2}$, $\hat{n}_1$, $\hat{n}_2$ and $\hat{p}$ be compatible chain or cycle representatives of the relative fundamental classes of these manifolds. $\exp 2\pi i \langle \hat{p}, \hat{h} \rangle$ is a complex number that depends on the chains $\hat{m}_i$ and $\hat{n}_i$. From the way it transforms under a change of $\hat{m}_i$ and $\hat{n}_i$, we see that it defines an invariant section of the functor $\mathcal{F}_2$ above, hence an element $\textcal{Int}(\check{c}, M^{d,2})$ of the object $\textcal{Int}(\check{c}, -N_1 \sqcup N_2)$ of the category $\textcal{Int}(\check{c}, -M_1 \sqcup M_2)$. This is indeed a 2-morphism in $\mathcal{H}_2$ from the 1-morphism $\textcal{Int}(\check{c}, N_1)$ to the 1-morphism $\textcal{Int}(\check{c}, N_2)$.

The second inverse limit works exactly as above and produces a 2-morphism $\mathcal{W\!\!Z}_c(M^{d,2})$ from the 1-morphism $\mathcal{W\!\!Z}_c(N_1)$ to the 1-morphism $\mathcal{W\!\!Z}_c(N_2)$, which depends only on the differential cohomology class $c$.

\paragraph{Compatibility with the dagger operation, monoidal structure and gluing} Recall that we have dagger operations on $\mathcal{B}^{d,2}_{\mathsf{F}}$, corresponding to orientation reversal, and on $\mathcal{H}_2$, corresponding to complex conjugation. The fact that $\mathcal{W\!\!Z}_c$ intertwines between the dagger operations is a direct consequence of the definition of the integral.

The fact that the pairing of chains and cochains is bilinear implies that $\mathcal{W\!\!Z}_c$ is multiplicative under disjoint unions of bordism, showing the compatibility with the monoidal structures of $\mathcal{B}^{d,2}_{\mathsf{F}}$ and $\mathcal{H}_2$.

Finally, the compatibility with the gluing follows directly from the locality of the integration functor.

\subsection{Wess-Zumino terms}

We would like to explain here how Wess-Zumino field theories, when seen as anomaly field theories, can be used to describe certain anomalous building blocks of physical quantum field theories, namely the Wess-Zumino terms. See also \cite{Fiorenza:2013jz} for a long list of examples.

\paragraph{General mechanism} Wess-Zumino terms in the physics literature appear as follows. Suppose that the structure $\mathsf{F}$ allows one to construct a differential cohomology class $c$ refining a characteristic class of degree $d+1$. For instance, if the characteristic class is a combination of Pontryagin classes of the tangent bundle, a Riemannian metric on the underlying manifold $M$, together with the associated Levi-Civita connection, allows one to construct a differential cocycle refining the characteristic class. More generally, if the characteristic class is a Chern-Weil class associated to a vector bundle $\mathcal{V}$ over $M$, a connection on $\mathcal{V}$ determines a differential cohomology class $c$ refining the Chern-Weil class. Typically, the structure $\mathsf{F}$ is such that the automorphism group ${\rm Aut}_{\mathsf{F}}(M)$ of an $\mathsf{F}$-manifold is non-trivial, but leaves the characteristic class invariant. In the two examples above, ${\rm Aut}_{\mathsf{F}}(M)$ corresponds to the subgroup of diffeomorphisms and gauge transformations of the connection on $\mathcal{V}$ that preserve the rest of the structure $\mathsf{F}$.

Once a differential refinement $c$ is available, we have a degree $d+1$ differential form $\omega_c$, the curvature of the differential cohomology class. We define on $d$-dimensional manifolds a Chern-Simons term ${\rm CS}_c$, which has the property that whenever $M^{d} = \partial M^{d+1}$,
\be
\label{EqDefCSBound}
{\rm CS}_c(M^d) = \int_{M^{d+1}} \omega_c \quad {\rm mod} \; 1\;.
\ee
Note that the fact that $\omega_c$ has integral periods ensures that ${\rm CS}_c(M^d)$ is well-defined modulo $1$. The completely general way of defining ${\rm CS}_c(M^{d+1})$ is through the integration map in differential cohomology \cite{hopkins-2005-70}. In the physics literature, ${\rm CS}_c(M^d)$ is often only defined under the assumption that the de Rham cohomology class of $\omega_c$ vanishes, by finding $\eta_c$ such that $d\eta_c = \omega_c$ and defining ${\rm CS}_c(M^d) = \int_{M^{d}} \eta_c$. Alternatively, it is defined on manifolds that bound via \eqref{EqDefCSBound}.

When considered on a manifold $M^{d,1}$ with boundary, ${\rm CS}_c(M^{d,1})$ is in general not invariant under ${\rm Aut}_{\mathsf{F}}(M^{d,1})$. A \emph{Wess-Zumino term} is a local quantity defined on $\partial M^{d,1}$ whose variation under ${\rm Aut}_{\mathsf{F}}(M^{d,1})$ is equal to the variation of ${\rm CS}_c(M^{d,1})$ modulo $1$. Such terms, being local, can be included in the Lagrangian of any classical field theory. Their anomalous variation under the symmetry group can cancel anomalous variations coming from quantum anomalies; this is the celebrated Green-Schwarz mechanism \cite{Green:1984sg}.

We now want to show that when a Wess-Zumino term is added to the action of a $d-1$-dimensional field theory $\mathcal{F}$, its anomaly field theory gets tensored with the corresponding $d$-dimensional Wess-Zumino field theory. To see this, remark that the partition function of the Wess-Zumino field theory is exactly the exponentiated Chern-Simons term $\exp 2 \pi i {\rm CS}_c(M^d)$. The variation of the partition function of the field theory $\mathcal{F}$ under an element of ${\rm Aut}_{\mathsf{F}}(M^{d-1})$ can be computed as follows. Let us pick any $d+1$-dimensional manifold $M^d$ with $\partial M^d = M^{d-1} \sqcup N$, for $N$ some $d-1$-dimensional manifold. We can always find such a pair $(N, M^d)$. Recall that the anomaly field theory $\mathcal{A}$, when evaluated on $M^d$, produces an element of the Hermitian line $\mathcal{A}(M^{d-1}) \otimes \mathcal{A}(N)$. Now consider the closed $d$-dimensional manifold $P_\phi$ obtained by gluing $M^d$ to $-M^d$, where we twist the gluing on $M^{d-1}$ by an automorphism $\phi \in {\rm Aut}_{\mathsf{F}}(M^{d-1})$. $\mathcal{A}(P_\phi)$ computes the inner product of $\mathcal{A}(M^{d})$ with $\phi \mathcal{A}(M^{d})$, i.e. the anomalous phase $\tau_{\phi, M^{d-1}}$ picked under $\phi$ by the partition function $\mathcal{F}(M^{d-1})$. If we add a Wess-Zumino term to the theory $\mathcal{F}$, by definition this phase will be multiplied by $\exp 2\pi i {\rm CS}_c(P_\phi)$. This shows that adding a Wess-Zumino term changes the anomaly field theory from $\mathcal{A}$ to $\mathcal{A} \otimes \mathcal{W\!\!Z}_c$.

\paragraph{Classical chiral WZW model} In the case of the (classical) two-dimensional WZW model, the famous Wess-Zumino term appearing in the action can be understood as follows in the framework above. Let $G$ be a semi-simple simply laced Lie group. The structure $\mathsf{F}$ contains the data of a principal $G$-bundle $\mathscr{E}$ with connection, in addition to an orientation, smooth structure and Riemannian metric. We write $F$ for the curvature of the connection. We have a degree 4 characteristic class, given by a linear combination of $c_2(\mathscr{E})$ and $c^2_1(\mathscr{E})$, whose de Rham cohomology class coincides with the de Rham cohomology class of ${\rm Tr}(F^2)$. The connection on $\mathscr{E}$ provides a differential refinement $c$. Then $\mathcal{W\!\!Z}_c$ is the standard three-dimensional classical Chern-Simons theory based on $G$. 

It is well-known that the variation of the value of the Chern-Simons action on a manifold with boundary under a gauge transformation is computed by the variation of the Wess-Zumino term of the classical WZW model evaluated at the boundary (see for instance Proposition 2 in Section 2.4 of \cite{2013arXiv1311.2490C}). The corresponding abelian bundle gerbes have also been extensively studied in the literature, see for instance \cite{Gawedzki:2002se,2005CMaPh.259..577C}.

\section{Dai-Freed theories}

\label{SecIndFieldTh}

Beyond Wess-Zumino terms, there is a class of anomalous theories of great physical interest, namely the chiral fermionic theories. The corresponding anomaly field theories are Dai-Freed theories, constructed in \cite{Dai:1994kq} as non-extended field theories. Our aim is to extend this construction down to codimension 2. Here we content ourselves with a definition of the corresponding field theory 2-functor $\mathcal{D\!\!F}_\mathsf{F}$, without checking explicitly that it satisfies all the axioms of a 2-functor.

\subsection{Definition} 

Assume that the data $\mathsf{F}(M)$ associated to a manifold $M$ includes a Riemannian metric and a vector bundle $\mathscr{V}$ with connection. The Riemannian metric allows us to turn the cotangent bundle of $M$ into a Clifford bundle, and we assume that $\mathscr{V}$ is a Clifford module endowed with a Clifford connection $\nabla^\mathscr{V}$. This data defines a Dirac operator $D_\mathsf{F}$ on $\mathscr{V}$. In a local coordinate frame, we have $D_\mathsf{F} = c(dx^\mu) \nabla^\mathscr{V}_{\partial_\mu}$, where $c$ is the map from $T^\ast M$ into the corresponding Clifford bundle. In the following, $d$ is odd.

\paragraph{Closed $d$-dimensional manifolds} On $M^d$, we can define the eta invariant $\eta_{\mathsf{F}}$ of $D_\mathsf{F}$. If $h_\mathsf{F}$ is the dimension of the kernel of $D_\mathsf{F}$, we define the modified eta invariant and the corresponding tau invariant by
\be
\xi_\mathsf{F} = \frac{\eta_{\mathsf{F}} + h_\mathsf{F}}{2} \;, \quad \tau_{\mathsf{F}} = \exp 2 \pi i \xi_\mathsf{F} \;.
\ee
Then 
\be
\mathcal{D\!\!F}_\mathsf{F}(M^d) = \tau_{\mathsf{F}} \;.
\ee

\paragraph{Closed $d-1$-dimensional manifolds} On $M^{d-1}$, $D_\mathsf{F}$ decomposes into two chiral Dirac operators $D_{\mathsf{F},+}$ and $D_{\mathsf{F},-}$. We can define a determinant line from the index vector space of $D_{\mathsf{F},+}$:
\be
\label{EqLineIndexFieldTheory}
L_\mathsf{F} := {\rm det}({\rm ker} D_{\mathsf{F},+}) \otimes \left({\rm det}({\rm coker}D_{\mathsf{F},+})\right)^{-1} \;,
\ee
where ${\rm det}$ denotes the top exterior power. We define 
\be
\mathcal{D\!\!F}_\mathsf{F}(M^{d-1}) = L_\mathsf{F} \;.
\ee

\paragraph{$d$-dimensional manifolds with boundary} Let us write $M = M^{d,1}$, $D = D_\mathsf{F}(M^{d,1})$ and $D_\partial = D_\mathsf{F}(\partial M^{d,1})$. Let $K_\partial^\pm(a)$ be the space of smooth positive/negative chirality eigenspinors of $D_\partial$ with eigenvalue smaller than $a > 0$. Let $E_\partial^\pm(a)$ be its complement. The Dirac Laplacian $D_\partial^2$ is invertible on $E_\partial^\pm(a)$. Let $T: K_\partial^+(a) \rightarrow K_\partial^-(a)$ be an isometry and consider the following boundary condition $B_{a,T}$ on a spinor $\psi$ on $M$:
\be
\label{EqBoundCondd1}
B_{a,T} \; : \; \psi_\partial^- + \left( T \oplus \left. \frac{D_\partial}{\sqrt{D_\partial^2}}\right|_{E_\partial^+(a)} \right) \psi_\partial^+ = 0\;, 
\ee
where $\psi_\partial^\pm$ are the positive/negative chirality components of the restriction of $\psi$ to $\partial M$. Dai and Freed \cite{Dai:1994kq} showed that $D$ with the boundary condition $B_{a,T}$ admits a well-defined eta invariant $\eta_{a,T}$. Taking an inverse limit to eliminate the dependence on the boundary condition, they show that the eta invariant becomes an element $\mathcal{D\!\!F}_\mathsf{F}(M^{d,1})$ of the determinant line $\mathcal{D\!\!F}_\mathsf{F}(\partial M^{d,1})$.

They also showed that the gluing relations are satisfied, and therefore that $\mathcal{D\!\!F}_\mathsf{F}: {\rm Bord}^{d,1}_\mathsf{F} \rightarrow \mathcal{H}_1$ is a functor, i.e. a field theory.

\paragraph{Closed $d-2$-dimensional manifolds} The value taken by the field theory on a closed $d-2$-dimensional manifold $M^{d-2}$ is a 2-Hermitian line associated to the index gerbe of the Dirac operator\footnote{We thank Dan Freed for suggesting the relevance of index gerbes in this context.} \cite{Segala, Carey:1997xm, 2002CMaPh.230...41L}. It can be constructed as follows. The data $\mathsf{F}$ determines again a Dirac operator $D_\mathsf{F}$ on $M^{d-2}$. Let $h \in C^\infty_c(\mathbb{R})$ be a smooth real-valued function with compact support. Using the spectral decomposition of $D_\mathsf{F}$, we can make sense of $h(D_\mathsf{F})$. Let us write $B$ for the subset of $C^\infty_c(\mathbb{R})$ consisting of functions such that the operator
\be
\label{EqChoFunch}
D_{\mathsf{F},h} := D_\mathsf{F} - h(D_\mathsf{F})
\ee
is invertible. For $h \in B$, let us write $H^h_>$ ($H^h_<$) for the space of smooth spinor fields generated by the eigenvectors of $D_{\mathsf{F},h}$ with positive (negative) eigenvalues. To any pair $h_1, h_2 \in B$, we can associate a Hermitian line
\be
\label{EqDefLinMorphIF}
L_{(h_1,h_2)} := {\rm det}(H^{h_1}_> \cap H^{h_2}_<) \otimes {\rm det}(H^{h_1}_< \cap H^{h_2}_>)^{-1}
\ee
where we define the determinant of the zero vector space to be $\mathbb{C}$. These lines can be used to construct a gerbe as follows. Consider the category $\mathcal{C}$ whose objects are maps $L: B \rightarrow \mathcal{T}_1$ such that 
\be
\label{EqCondElIF}
L(h_1) = L_{(h_1,h_2)} \otimes L(h_2) \;,
\ee
and whose morphisms are functors $\mathcal{T}_1 \rightarrow \mathcal{T}_1$ that preserve these relations. We are here freely identifying Hermitian lines with $\mathbb{T}$-torsors, see Appendix \ref{AppCircGerbes}. $\mathcal{C}$ is a $\mathcal{T}_1$-torsor. Indeed, remark that if we pick a particular $h$, we can make an arbitrary choice for the value $L(h)$. Once this choice is made, the map $L$ is fully determined by \eqref{EqCondElIF}. This provides a (non-canonical) equivalence between $\mathcal{C}$ and $\mathcal{T}_1$. Moreover, as the map takes value in $\mathcal{T}_1$, the product on $\mathcal{T}_1$ provides a free transitive action of $\mathcal{T}_1$ on $\mathcal{C}$, which is therefore a $\mathcal{T}_1$-torsor. Using \eqref{EqDef2HermLineFromTGerbe}, we define 
\be
\mathcal{D\!\!F}_\mathsf{F}(M^{d-1}) = \mathcal{L}_\mathcal{C} \;,
\ee
the 2-Hermitian line associated to $\mathcal{C}$. 

\paragraph{$d-1$-dimensional manifold with boundary} Consider now a $d-1$-dimensional manifold with boundary $M^{d-1,1}$. We want to check that $\mathcal{D\!\!F}_\mathsf{F}(M^{d-1,1}) \in \mathcal{D\!\!F}_\mathsf{F}(\partial M^{d-1,1})$.
To see this, let us first simplify the notation and write $M := M^{d-1,1}$, $D_+ := D_{\mathsf{F},+}(M^{d-1,1}) $ and $D_\partial := D_{\mathsf{F}}(\partial M^{d-1,1})$. Recall that the Riemannian metric on $M$ is isometric to a product metric on a neighborhood $N$ of the boundary. 
Writing $t$ for the coordinate normal to the boundary, we have on $N$ $D_+ = c(dt)\nabla^\mathscr{V}_{\partial_t} + D_\partial$. The standard Atiyah-Patodi-Singer (APS) boundary conditions require spinors of positive (negative) chirality on $M$ to restrict to $H^{0}_<$ ($H^{0}_>$) on the boundary, where we use the same notation as in the previous paragraph, with $M^{d-2} = \partial M$. We consider the more general boundary conditions where the positive (negative) chirality spinors are required to restrict in $H^{h}_<$ ($H^{h}_>$), for $h \in B$. These boundary conditions differ from the APS boundary conditions on a finite dimensional subspace and are elliptic as well (see Chapter 18 of \cite{booss1993elliptic}). Then for each choice $h$ of boundary conditions, we obtain a line 
\be
L_\mathsf{F}(M,h) := {\rm det}({\rm ker} D_+) \otimes \left({\rm det}({\rm coker}D_+)\right)^{-1} \;,
\ee
where it is understood that the right-hand side is computed with the boundary condition $h$. We do not want to pick a preferred boundary condition, so we should think of $\mathcal{D\!\!F}_\mathsf{F}(M)$ as a function associating the line $L_\mathsf{F}(M,h)$ to each $h \in B$. We now prove:
\begin{lemma}
$\mathcal{D\!\!F}_\mathsf{F}(M) \in \mathcal{D\!\!F}_\mathsf{F}(\partial M)$
\end{lemma}
\begin{proof}
We must prove that the lines $L_\mathsf{F}(M,h)$, $h \in B$, satisfy \eqref{EqCondElIF}. 
For that, we need to understand the effect of a change of the boundary condition from $h_1$ to $h_2$. Let $W_+$ ($W_-$) be the space of smooth spinors on $\partial M$ extending to $M$ as smooth positive (negative) chirality spinors $\psi$ solving the Dirac equation $D_+ \psi = 0$ ($D_- \psi = 0$). We have $W_+ \cap W_- = \{0\}$, from the invertibility of the Dirac operator on the double (\cite{booss1993elliptic}, Chapter 9). We also have
\be
K_{h,+} := {\rm ker} D_+ \simeq W_+ \cap H^h_< \;, \quad K_{h,-} := {\rm coker}D_+ \simeq W_- \cap H^h_> \;,
\ee
where the isomorphisms are given by restriction to the boundary, which is surjective by definition. The injectivity follows from the unique continuation property of the Dirac operator (\cite{booss1993elliptic}, Chapter 8), which ensures that there are no non-trivial solutions of the Dirac equation that restrict trivially on the boundary. Now consider the virtual vector space
\be
V = (H^{h_1}_> \cap H^{h_2}_<) \ominus (H^{h_1}_< \cap H^{h_2}_>) \;.
\ee
Note that ${\rm det V} = L_{(h_1,h_2)}$. We have
\be
V \cap W_+ \simeq K_{h_2,+}/(K_{h_1,+} \cap K_{h_2,+}) \ominus K_{h_1,+}/(K_{h_1,+} \cap K_{h_2,+})
\ee
\be
V \cap W_- \simeq K_{h_1,-}/(K_{h_1,-} \cap K_{h_2,-}) \ominus K_{h_2,-}/(K_{h_1,-} \cap K_{h_2,-})
\ee
The existence of the solution of the Dirichlet problem for the associated Dirac Laplacian ensures that $(V \cap W_+) \oplus (V \cap W_-) = V$. By taking the direct sum of the previous two equations and then taking the determinant, we get
\be
L_{(h_1,h_2)} = L_\mathsf{F}(M,h_2) \otimes (L_\mathsf{F}(M,h_1))^{-1}
\ee
so \eqref{EqCondElIF} is satisfied.
\end{proof}

\paragraph{$d$-dimensional manifolds with corners} Assume that $M^{d,2}$ has boundary components $N_1$ and $N_2$, and that the latter both have boundary $-M_1 \sqcup M_2$. A choice of function $h$ as in \eqref{EqChoFunch} provides boundary conditions for the Dirac operators on $N_1$ and $N_2$. Once we have elliptic Dirac operators on the boundary, we can pick boundary conditions for the Dirac operator on $M^{d,2}$ as in \eqref{EqBoundCondd1}. The whole boundary condition depends on a triplet $(h,a,T)$. To construct the theory rigorously on $d$-dimensional manifolds with corners, one should show that the Dirac operator on $M^{d,2}$ with these boundary conditions admits a well-defined eta invariant, for instance along the lines of Appendix A of \cite{Dai:1994kq}. Taking an inverse limit to eliminate the dependence on the boundary condition, one should obtain an element $\mathcal{D\!\!F}_\mathsf{F}(M^{d,2})$ of the object $\mathcal{D\!\!F}_\mathsf{F}(-N_1 \sqcup N_2)$ of the category $\mathcal{D\!\!F}_\mathsf{F}(-M_1 \sqcup M_2)$.

\paragraph{Compatibility with the dagger operation, monoidal structure and gluing} We know from \cite{Dai:1994kq} that these compatibility conditions are satisfied on $d$- and $d-1$-dimensional manifolds. On a $d-2$-dimensional manifold $M^{d-2}$, a flip of the orientation multiplies the Dirac operator by $-1$. For consistency, we must also change the function $h$ to $\bar{h}(x) = -h(-x)$. We see then that $L_{(h_1,h_2)}(M^{d-2}) = \left(L_{(\bar{h}_1,\bar{h}_2)}(-M^{d-2})\right)^{-1}$. After taking the inverse limit, we see that $\mathcal{D\!\!F}_\mathsf{F}(-M^{d-2}) = \left( \mathcal{D\!\!F}_\mathsf{F}(M^{d-2}) \right)^\dagger$.

The compatibility with the monoidal structure comes readily from the fact that the spectrum of a Dirac operator on a manifold with several connected components is the union of the spectra of the restrictions to each component. This implies in particular that the 2-Hermitian line associated to the whole manifold is the tensor product of the 2-Hermitian lines associated to the components.

The gluing condition seems more difficult to check and we will not attempt this here.

\subsection{Relation to chiral fermionic theories}

On $d-1$-dimensional manifolds, the structure $\mathsf{F}$ determines a Dirac operator $D$ that decomposes into two chiral Dirac operators $D_+$ and $D_-$. The chiral Dirac operators can be used to construct chiral fermionic field theories that are generally anomalous. It is known that over the moduli space of $\mathsf{F}$-structures, the partition function of a chiral fermionic theory is a section of the determinant line bundle of the chiral Dirac operator \cite{Freed:1986zx}. The state space is an abelian bundle gerbe over the moduli space constructed from the index gerbe \cite{Segala, Carey:1997xm}. Our definition of the cobordism category does not allow us to speak about moduli spaces of cobordisms, but we can restrict the statements above to a single point in the moduli space. We therefore see that the partition function is an element of a determinant line and that the state space is an abelian index gerbe. These facts are naturally explained if the anomaly field theory of the chiral fermionic field theory is the extended Dai-Freed theory constructed above.

\subsection*{Acknowledgments}

I would like to thank Dan Freed for useful correspondence and for pointing out a problem with the definition of the symmetry groups in a previous version of this paper. This research is supported in part by Forschungskredit FK-14-108, SNF Grant No.200020-149150/1 and by NCCR SwissMAP, funded by the Swiss National Science Foundation.

\appendix

\section{Review of some relevant concepts}

\subsection{2-categories and 2-functors}

\label{App2cat2funct}

A useful reference for what follows is \cite{1998math.....10017L}.

\paragraph{2-categories} A \emph{strict 2-category} $\mathcal{C}$ is a category enriched in categories, namely a category such that the collection of morphisms between any two objects is itself a category. In more detail, it consists of the following:
\begin{itemize}
\item A collection $\mathcal{O}_{\mathcal{C}}$ of objects.
\item For each pair $X,Y \in \mathcal{O}_{\mathcal{C}}$, a category $\textcal{Hom}_{\mathcal{C}}(X,Y)$ of morphisms. The objects of $\textcal{Hom}_{\mathcal{C}}(X,Y)$ are called 1-morphisms and the morphisms in the category $\textcal{Hom}_{\mathcal{C}}(X,Y)$ are called 2-morphisms.
\item For each triplets $X,Y,Z \in \mathcal{O}_{\mathcal{C}}$, a composition functor $\textcal{Hom}_{\mathcal{C}}(X,Y) \times \textcal{Hom}_{\mathcal{C}}(Y,Z) \rightarrow \textcal{Hom}_{\mathcal{C}}(X,Z)$.
\item For each $X \in \mathcal{O}_{\mathcal{C}}$ a 1-morphism ${\rm id}_X \in \textcal{Hom}_{\mathcal{C}}(X,X)$ that acts as a unit with respect to the composition.
\item The composition is required to be strictly associative, namely for $W,X,Y,Z \in \mathcal{O}_{\mathcal{C}}$, the two obvious functors mapping $\textcal{Hom}_{\mathcal{C}}(W,X) \times \textcal{Hom}_{\mathcal{C}}(X,Y) \times \textcal{Hom}_{\mathcal{C}}(Y,Z)$ to $\textcal{Hom}_{\mathcal{C}}(W,Z)$ (obtained by composing two composition functors) coincide.
\end{itemize}

In a \emph{weak 2-category}, or a \emph{bicategory}, the unit ${\rm id}_X$ is only required to satisfy the unit axiom up to a 2-isomorphism (i.e. up to an invertible morphism in the appropriate morphism category). Similarly, the composition functors need to satisfy the associativity conditions only up to 2-isomorphisms. The corresponding diagrams can be found in \cite{1998math.....10017L}.

\paragraph{2-functors} Let $\mathcal{C}$ and $\mathcal{D}$ be 2-categories. A \emph{2-functor} $\mathcal{F}$ between $\mathcal{C}$ and $\mathcal{D}$ consists of
\begin{itemize}
\item A map $\mathcal{F}_\mathcal{O}: \mathcal{O}_\mathcal{C} \rightarrow \mathcal{O}_\mathcal{D}$.
\item For each $X,Y \in \mathcal{O}_\mathcal{C}$, a functor $\mathcal{F}_{X,Y}: \textcal{Hom}_{\mathcal{C}}(X,Y) \rightarrow \textcal{Hom}_{\mathcal{D}}(\mathcal{F}_\mathcal{O}(X),\mathcal{F}_\mathcal{O}(Y))$. This functor has to intertwine the composition of morphisms in $\mathcal{C}$ and $\mathcal{D}$, and preserve the units.
\end{itemize}
Depending on which type of 2-categories we are working with, "intertwine" and "preserve" are understood either exactly or up to natural transformations. Again, see \cite{1998math.....10017L} for details.

\paragraph{2-natural transformations} A \emph{2-natural transformation} $\textcal{n}$ between two 2-functors $\mathcal{F}, \mathcal{G}$ between two 2-categories $\mathcal{C},\mathcal{D}$ consists of:
\begin{itemize}
\item For each $X \in \mathcal{O}_\mathcal{C}$, an object $\textcal{n}(X) \in \textcal{Hom}_{\mathcal{D}}(\mathcal{F}(X),\mathcal{G}(X))$.
\item For each $X,Y \in \mathcal{O}_\mathcal{C}$, $f \in \textcal{Hom}_{\mathcal{C}}(X,Y)$, a morphism $\textcal{n}(f)$ of the category $\textcal{Hom}_{\mathcal{D}}(\mathcal{F}(X),\mathcal{G}(Y))$ from $\mathcal{G}(f) \circ \textcal{n}(X)$ to $\textcal{n}(Y) \circ \mathcal{F}(f)$.
\end{itemize}
These morphisms must satisfy relations that are spelled out in \cite{1998math.....10017L}.

\subsection{2-vector spaces and 2-Hilbert spaces}

\label{App2Vect}

\paragraph{2-vector spaces} 2-vector spaces were first defined in \cite{MR1278735} (see also \cite{Yetter1993}), but we follow here the approach of \cite{2008arXiv0812.4969B}, Section 3.2. A category $\mathcal{C}$ is $\mathbb{C}$\emph{-linear} if for each pair of objects $X,Y \in \mathcal{C}$, the collection of morphisms ${\rm Hom}_{\mathcal{C}}(X,Y)$ from $X$ to $Y$ is a finite dimensional complex vector space, and the composition of morphisms is bilinear. A \emph{linear functor} between $\mathbb{C}$-linear categories $\mathcal{C}, \mathcal{D}$ is a functor $\mathcal{F}: \mathcal{C} \rightarrow \mathcal{D}$ such that the induced map $\mathcal{F}_{X,Y}: {\rm Hom}_{\mathcal{C}}(X,Y) \rightarrow {\rm Hom}_{\mathcal{D}}(\mathcal{F}(X),\mathcal{F}(Y))$ is a $\mathbb{C}$-linear map. (A word of warning, what is called a linear functor in \cite{2008arXiv0812.4969B} is what we call a $\mathcal{V}_1$-linear functor, see the definition below.) A \emph{linear equivalence} of $\mathbb{C}$-linear categories is an equivalence of categories whose underlying functors are linear. Finally, define $\mathcal{V}_1^n$ to be the $n$th Cartesian product of the category $\mathcal{V}_1$ of finite vector spaces.

A \textit{finite dimensional 2-vector space} is a $\mathbb{C}$-linear category linearly equivalent to $\mathcal{V}_1^n$ for some $n \in \mathbb{N}$.

To make sense of this definition, one should think of the category $\mathcal{V}_1$ as taking the role that $\mathbb{C}$ is playing for finite-dimensional vector spaces. While $\mathbb{C}$ is a field, we only have a semiring structure on $\mathcal{V}_1$ provided by the direct sum and the tensor product (i.e. $\mathcal{V}_1$ is a symmetric bimonoidal category). $\mathcal{V}_1^n$ is a free module over the category $\mathcal{V}_1$ with $n$ generators. An obvious consequence of the definition above is that any 2-vector space $\mathcal{C}$ can be pictured as $\mathcal{V}_1^n$, albeit non-canonically. We call $n$ the dimension of the 2-vector space. This is the analog of the fact that any complex vector space can be pictured as $\mathbb{C}^n$, generally in a non-canonical way. We can extend the operation of direct sum componentwise to $\mathcal{V}_1^n$. The categorical biproduct provides a monoidal structure on $\mathcal{C}$, which coincides with the componentwise direct sum under the equivalence with $\mathcal{V}_1^n$. This is the analog of the addition operation on vector spaces. Scalar multiplication of an object $O \in \mathcal{V}_1^n$ by $V \in \mathcal{V}_1$ is defined by taking the tensor product of $V$ with each of the components of $O$. This induces a scalar multiplication on $\mathcal{C}$ up to isomorphism. There is also a zero-dimensional 2-vector space equivalent to $\mathcal{V}_1^0$, that has a unique object and morphism. In general there is no notion of tensor product in a 2-vector space $\mathcal{C}$, just like there is no notion of product on a generic vector space.

\paragraph{The 2-category $\mathcal{V}_2$} The collection $\textcal{Cat}$ of all categories can be given the structure of a strict 2-category, whose objects are categories, whose 1-morphisms are functors and whose 2-morphisms are natural transformations. We will construct the 2-category $\mathcal{V}_2$ of 2-vector spaces as a subcategory of $\textcal{Cat}$. 

Given any two 2-vector spaces $\mathcal{C} \sim \mathcal{V}_1^n$ and $\mathcal{D} \sim \mathcal{V}_1^m$, a functor $\mathcal{F}: \mathcal{C} \rightarrow \mathcal{D}$ determines a naturally isomorphic functor $\tilde{\mathcal{F}}: \mathcal{V}_1^n \rightarrow \mathcal{V}_1^m$. A \textit{$\mathcal{V}_1$-linear functor} is a functor $\mathcal{F}$ such that $\tilde{\mathcal{F}}$ is compatible with the $\mathcal{V}_1$-module structures on $\mathcal{V}_1^n$ and $\mathcal{V}_1^m$. $\tilde{\mathcal{F}}$ always takes the form of an $m \times n$ matrix of complex vector spaces, acting on $\mathcal{V}_1^n$ by the usual rules of matrix multiplication \cite{2008arXiv0812.4969B}. The 1-morphisms in $\mathcal{V}_2$ are $\mathcal{V}_1$-linear functors between 2-vector spaces and the 2-morphisms are natural transformations.
 
The higher analogues of the direct sum and tensor product should be a pair of commutative monoidal structures on $\mathcal{V}_2$ satisfying the axioms of a semiring. $\mathcal{V}_2$ is an additive category, and the direct sum is provided by the categorical biproduct. One can check that given two 2-vector spaces $\mathcal{C}$ and $\mathcal{D}$, their biproduct $\mathcal{C} \oplus \mathcal{D}$ is a 2-vector space as well. Indeed, given two linear equivalences $\mathcal{C} \rightarrow \mathcal{V}_1^n$ and $\mathcal{D} \rightarrow \mathcal{V}_1^m$, one can construct a linear equivalence $\mathcal{C} \oplus \mathcal{D} \rightarrow \mathcal{V}_1^{n+m}$.

For a coordinate independent description of the tensor product, we refer the reader to Section 4.4 of \cite{1996q.alg.....9018B}. Given two 2-vector spaces $\mathcal{C}$ and $\mathcal{D}$ endowed with linear equivalences $\mathcal{C} \rightarrow \mathcal{V}_1^n$ and $\mathcal{D} \rightarrow \mathcal{V}_1^m$, $\mathcal{C} \otimes \mathcal{D}$ is linearly equivalent to $\mathcal{V}_1^{nm}$, as the intuition from the tensor product of vector spaces suggests. Given objects in $\mathcal{C}$ and $\mathcal{D}$, pictured as vectors of size $n$ and $m$ of vector spaces, their tensor product in $\mathcal{V}_1^{nm}$ is an $n$ by $m$ matrix of vector spaces whose entries are the (ordinary) tensor products of the vector space components. The tensor product is defined similarly on morphisms between objects. It also extends to 1-morphisms and 2-morphisms in $\mathcal{V}_2$.

\paragraph{2-Hilbert spaces} We now turn to the definition of finite dimensional 2-Hilbert spaces. A comprehensive reference is \cite{1996q.alg.....9018B}.

Let us write $\mathcal{H}_1$ for the category of finite dimensional Hilbert spaces. Following the same logic as above, in the realm of 2-Hilbert spaces, the role $\mathbb{C}$ is playing for finite dimensional Hilbert spaces should be taken over by $\mathcal{H}_1$. The role of the inner product will be played by the hom functor. Recall that given any (locally small) category $\mathcal{C}$ the hom functor is a functor ${\rm Hom}: \mathcal{C}^{\rm op} \times \mathcal{C} \rightarrow \textcal{Set}$ taking a pair of objects $(X,Y)$ to the set of morphisms ${\rm Hom}(X,Y)$ between $X$ and $Y$. Suppose now that $\mathcal{C}$ is a 2-vector space. As $\mathcal{C}$ is $\mathbb{C}$-linear, ${\rm Hom}(X,Y)$ is a complex vector space. In order to ensure that ${\rm Hom}(X,Y)$ is a Hilbert space, we need to restrict ourselves to the 2-vector spaces that are categories enriched in $\mathcal{H}_1$, i.e. such that their vector spaces of morphisms are Hilbert spaces. In addition, we must ensure that the inner product is Hermitian. The analog of the complex conjugation on $\mathbb{C}$ is the complex conjugation of Hilbert spaces in $\mathcal{H}_1$. We therefore need isomorphisms ${\rm Hom}(X,Y) \simeq \overline{{\rm Hom}(Y,X)}$. This happens if $\mathcal{C}$ is an $H^\ast$-category \cite{1996q.alg.....9018B}. Practically, an $H^\ast$-category is a category enriched over $\mathcal{H}_1$ equipped with antilinear maps $\dagger: {\rm Hom}(X,Y) \simeq {\rm Hom}(Y,X)$ satisfying:
\begin{itemize}
\item $f^{\dagger\dagger} = f$
\item $(fg)^\dagger = g^\dagger f^\dagger$
\item $\langle fg, h \rangle = \langle g, f^\dagger h \rangle$
\item $\langle fg, h \rangle = \langle f, h g^\dagger \rangle$
\end{itemize}
for all $f \in {\rm Hom}(X,Y)$, $g \in {\rm Hom}(W,X)$, $h \in {\rm Hom}(W,Y)$, $W,X,Y \in \mathcal{O}_\mathcal{C}$. $\langle \bullet, \bullet \rangle$ denotes here the inner products on the hom Hilbert spaces. Note that the first two axioms are those of a $\dagger$-structure on $\mathcal{C}$.

Therefore, to summarize, a 2-Hilbert space is a 2-vector space that is also an $H^\ast$-category. The inner product is valued in $\mathcal{H}_1$, the category of finite dimensional Hilbert space. It is sesquilinear with respect to the scalar multiplication by elements of $\mathcal{H}_1$ and the complex conjugation of Hilbert spaces. We write $\mathcal{H}_2$ for the 2-category of 2-Hilbert spaces. 

The tensor product on the category of 2-vector spaces passes to $\mathcal{H}_2$. Let $\mathcal{C}_1$ and $\mathcal{C}_2$ be two 2-Hilbert spaces and let $C_1, C_3 \in \mathcal{C}_1$ and $C_2, C_4 \in \mathcal{C}_2$. Then we have the following relation between the inner products:
\be
\langle C_1 \otimes C_2, C_3 \otimes C_4 \rangle_{\mathcal{C}_1 \otimes \mathcal{C}_2} = \langle C_1, C_3 \rangle_{\mathcal{C}_1} \otimes \langle C_2, C_4 \rangle_{\mathcal{C}_2} \;,
\ee
where the tensor product on the right-hand side is the one in $\mathcal{H}_1$. This formula mimics of course the properties of the tensor product in $\mathcal{H}_1$ with respect to the multiplication, one degree higher in the category hierarchy.

It will sometimes be convenient to see $\mathbb{C}$ as a category $\mathcal{H}_0$  that has no morphisms and whose objects are elements of $\mathbb{C}$.

\paragraph{2-Hermitian lines} Let us first discuss Hermitian lines. A Hermitian line $L$ is a Hilbert space isomorphic to $\mathbb{C}$. What makes this concept non-vacuous is that there might not be a canonical isomorphism to $\mathbb{C}$. This is what happens for instance with the determinant line of a Dirac operator (see Section \ref{SecIndFieldTh}). A concrete consequence of the fact that the determinant line is not canonically isomorphic to $\mathbb{C}$ is that when considered in families, the determinant line becomes a possibly non-trivial determinant line bundle. If there was a canonical isomorphism from the determinant line to $\mathbb{C}$, this line bundle would necessarily be trivial.

The definition of 2-Hermitian lines is completely analogous. A 2-Hermitian line $\mathcal{L}$ is a one-dimensional 2-Hilbert space, and therefore a category equivalent to $\mathcal{H}_1$, but not necessarily canonically. We will see a way of constructing 2-Hermitian lines in the next section. The tensor product of two 2-Hermitian lines is again a 2-Hermitian line, as is easily checked by picking an equivalence with $\mathcal{H}_1$ and using the definition of the tensor product given above. Moreover, every 2-Hermitian line $\mathcal{L}$ admits an inverse $\mathcal{L}^{-1}$ such that $\mathcal{L} \otimes \mathcal{L}^{-1}$ is canonically equivalent to $\mathcal{H}_1$.

\subsection{Higher circle groups}

\label{AppCircGerbes}

\paragraph{$\mathbb{T}$-torsors} Consider the unit circle group $\mathbb{T} \subset \mathbb{C}$. Its higher categorical analogues can be described as follows (see \cite{Freed:1994ad}, Section 1). Let us define a \textit{$\mathbb{T}$-torsor} $T$ to be a manifold endowed with a smooth, free and transitive action of $\mathbb{T}$. As the action of $\mathbb{T}$ is free and transitive, $T$ is diffeomorphic to $\mathbb{T}$, but generally not in a canonical way. We write the action of $\mathbb{T}$ on $T$ as $t \cdot \tau$, for $t \in T$ and $\tau \in \mathbb{T}$. Remark that we can canonically associate a $\mathbb{T}$-torsor to a Hermitian line and vice versa. Given a $\mathbb{T}$-torsor $T$, we can define
\be
\label{EqHermLineFromTors}
L_T = \{(t, z) \in T \times \mathbb{C}\}/\{(t \cdot \tau, z) \sim (t, \tau z), \tau \in \mathbb{T} \} \;.
\ee 
Conversely, given a Hermitian line $L$, its unit norm elements form a $\mathbb{T}$-torsor. Because of this, we will often not distinguish explicitly between $\mathbb{T}$-torsor and Hermitian lines.

Write $\mathcal{T}_0 = \mathbb{T}$ and let $\mathcal{T}_1$ be the category of $\mathbb{T}$-torsors, the morphisms being smooth maps intertwining the actions of $\mathbb{T}$. Remark that due to the identification of Hermitian lines with $\mathbb{T}$-torsors, we can see $\mathcal{T}_1$ as a subcategory of $\mathcal{H}_1$, namely the subcategory of Hermitian lines. $\mathbb{T}$ is itself a $\mathbb{T}$-torsor, corresponding to the Hermitian line $\mathbb{C}$. Given two torsors $T_1$, $T_2$, the space of morphisms from $T_1$ to $T_2$ in $\mathcal{T}_1$ is a $\mathbb{T}$-torsor as well. Moreover, any morphism is invertible, making $\mathcal{T}_1$ a groupoid. Just like there is a group structure on $\mathbb{T}$, there is a (weak) 2-group structure on $\mathcal{T}_1$. Given two torsors $T_1$ and $T_2$, we define the multiplication
\be
\label{EqDefProdTorsors}
T_1 \cdot T_2 = \{(t_1, t_2) \in T_1 \times T_2 \} / \{(t_1 \cdot \tau, t_2) \sim (t_1, t_2 \cdot \tau), \tau \in \mathbb{T} \}
\ee
This multiplication is associative up to canonical isomorphism. The identity is $\mathbb{T}$, and $T^{-1}$ is the torsor that coincides with $T$ as a manifold and carries the action $t \cdot_{T^{-1}} \tau = t \cdot_{T} \tau^{-1}$, $t \in T$, $\tau \in \mathbb{T}$. There is a canonical isomorphism between $T \cdot T^{-1}$ and $\mathbb{T}$.

\paragraph{$\mathbb{T}$-gerbes} Moving one step higher in the categorical hierarchy, we define a $\mathbb{T}$-gerbe $\mathcal{G}$ to be a $\mathcal{T}_1$-torsor. By this, we mean that $\mathcal{G}$ is a category endowed with a free and transitive action of $\mathcal{T}_1$. This action is described by a functor $\mathcal{G} \times \mathcal{T}_1 \rightarrow \mathcal{G}$, $(G, T) \mapsto G \cdot T$. We can ensure that it is free and transitive by requiring the existence of an equivalence of categories between $\mathcal{G} \times \mathcal{T}_1$ and $\mathcal{G} \times \mathcal{G}$ mapping $(G, T)$ to $(G, G \cdot T)$. We refer the reader to \cite{Freed:1994ad} for a bit more information about $\mathcal{T}_1$-torsors. We only note here that they form a 2-category $\mathcal{T}_2$.

There is a canonical bijection between $\mathbb{T}$-gerbes and 2-Hermitian lines. Given a $\mathbb{T}$-gerbe $\mathcal{G}$, we can construct the 2-Hermitian line
\be
\label{EqDef2HermLineFromTGerbe}
\mathcal{L}_\mathcal{G} = \{(G, H) \in \mathcal{G} \times \mathcal{H}_1\}/\{(G \cdot T, H) \sim (G, H \otimes L_T), T \in \mathcal{T}_1\}
\ee
(compare with \eqref{EqHermLineFromTors}). Conversely, given a 2-Hermitian line $\mathcal{L}$, consider the subcategory $\mathcal{G}$ formed by all of the objects that mapped to $\mathcal{T}_1 \subset \mathcal{H}_1$ by any (non-canonical) isomorphism $\mathcal{L} \simeq \mathcal{H}_1$. $\mathcal{G}$ is a $\mathbb{T}$-gerbe that is independent of the choice of isomorphism used to define it. Because of this bijection, we will not always distinguish between $\mathbb{T}$-gerbes and 2-Hermitian lines.

\subsection{Geometric bordism 2-categories}

\label{ApGeomBord2Cat}

We sketch here the construction of the geometric bordism 2-categories that are the domains of the functors representing the various anomaly field theories. A detailed treatment of geometric bordism categories can be found in \cite{Ayal}. We adopt an approach in which bordisms are defined abstractly and which yields a strict 2-category. In order to treat families, it would be best to used framed bordisms, i.e. to picture manifolds and bordisms as embedded in $\mathbb{R}^n$ for some large $n$ together with a trivialization of their normal bundle. Indeed, in this case the moduli spaces of bordisms and manifolds come with natural topologies. In this approach one would presumably only obtain a weak 2-category.

\paragraph{Manifolds with structures} We will assume that we have a geometric/topological structure $\mathsf{F}$ that can be put on smooth manifolds of any dimension.  We call manifolds endowed with an $\mathsf{F}$-structure $\mathsf{F}$-manifolds. We list below a series of assumptions that $\mathsf{F}$ should satisfy and that are fulfilled by the concrete examples of such structures met in the main text. The correct formalization is probably the concept of equivariant sheaf of \cite{Ayal}.   

We assume that $\mathsf{F}$ always includes an orientation and a smooth structure. We assume that we have a well-defined notion of germ of $\mathsf{F}$-structure on submanifolds. In the following, a \emph{codimension} $p$ \emph{germ} of $\mathsf{F}$-structure on a manifold $M$ is a germ for the inclusion of $M \times \{0\}  \subset M \times (-\epsilon, \epsilon)^p$ for some $\epsilon \in \mathbb{R}_+$. 

We also assume that we can pull-back (germs of) $\mathsf{F}$-structures along smooth maps of manifolds. We define a morphism of $\mathsf{F}$-manifolds to be a smooth map of the underlying manifolds that preserves the (germs of) $\mathsf{F}$-structure, namely the pulled back (germ of) $\mathsf{F}$-structure should coincide with the (germ of) $\mathsf{F}$-structure on the domain. $\mathsf{F}$-manifolds, possibly endowed with germs of $\mathsf{F}$-structure, then form a category $\mathcal{M}_\mathsf{F}$. We write ${\rm Aut}_\mathsf{F}(M)$ for the automorphism group of $M \in \mathcal{M}_\mathsf{F}$. If $\phi$ is an isomorphism in $\mathcal{M}_\mathsf{F}$ with source $M$, we write $\phi M$ for the target object. 


\paragraph{Killing automorphisms} As shown Section \ref{SecAnPartFunc}, the existence of automorphisms in $\mathcal{M}_\mathsf{F}$ can sometimes prevent the existence of non-trivial natural transformations representing anomalous field theories. We therefore describe here a way of killing part or all of these automorphism groups. This is achieved by constructing a category whose objects are manifolds endowed with extra structures that are not necessarily preserved by the morphisms.

Suppose that we are given a structure $\mathsf{E}$ refining the structure $\mathsf{F}$, in the sense that an $\mathsf{E}$-structure on a manifold $M$ determines an $\mathsf{F}$-structure on $M$. More precisely, we assume that there is a "forgetful" functor $\textcal{o}\,_{\mathsf{E},\mathsf{F}}: \mathcal{M}_\mathsf{E} \rightarrow \mathcal{M}_\mathsf{F}$ that is surjective on objects. We often write $\mathsf{F}$-manifold between square brackets, and we write $[M]$ for $\textcal{o}\,_{\mathsf{E},\mathsf{F}}(M)$. 

We want now to define a category $\mathcal{M}_{\mathsf{E},\mathsf{F}}$ whose objects are manifolds endowed with $\mathsf{E}$-structure, but whose morphisms preserve only the associated $\mathsf{F}$-structure. For that, we need to know how the morphisms of $\mathsf{F}$-structures act on $\mathsf{E}$-structures. We can encode this data in a function $h_\sigma: \textcal{o}\,_{\mathsf{E},\mathsf{F}}^{-1}([M_2]) \rightarrow \textcal{o}\,_{\mathsf{E},\mathsf{F}}^{-1}([M_1])$ for each morphism $\sigma: [M_1] \rightarrow [M_2]$ in $\mathcal{M}_\mathsf{F}$, such that 
\begin{itemize}
\item $h_{{\rm id}_{[M]}} = {\rm id}_{\textcal{o}\,_{\mathsf{E},\mathsf{F}}^{-1}([M])}$ for any object $[M]$ in $\mathcal{M}_\mathsf{F}$,
\item $h_\sigma \circ h_\rho = h_{\rho \circ \sigma}$ for any pair of composable morphisms $\rho, \sigma$ in $\mathcal{M}_\mathsf{F}$.
\end{itemize}
We can now define $\mathcal{M}_{\mathsf{E},\mathsf{F}}$ as the category whose objects coincide with those of $\mathcal{M}_\mathsf{E}$ and that has a morphism between $M_1, M_2 \in \mathcal{M}_{\mathsf{E},\mathsf{F}}$ for each $\sigma: [M_1] \rightarrow [M_2]$ such that $h_\sigma(M_2) = M_1$. Such a morphism $\tilde{\sigma}$ covers the same diffeomorphism as $\sigma$ and acts on the $\mathsf{E}$-structures so that $\tilde{\sigma}^\ast \mathsf{E}(M_2) = \mathsf{E}(M_1)$. The composition of such morphisms can be defined consistently thanks to the axioms for $h_\sigma$ above. We call the manifolds in $\mathcal{M}_{\mathsf{E},\mathsf{F}}$ $(\mathsf{E},\mathsf{F})$-manifolds. Remark that $h$ defines an action of the group ${\rm Aut}_\mathsf{F}([M])$ on the fiber $\textcal{o}\,_{\mathsf{E},\mathsf{F}}^{-1}([M])$, composed of manifolds with an $\mathsf{E}$-structure refining the $\mathsf{F}$-structure on $[M]$. 

Let us finally remark each morphism in $\mathcal{M}_{\mathsf{E},\mathsf{F}}$ covers a morphism of $\mathsf{F}$-manifolds, so that we have a forgetful functor from $\mathcal{M}_{\mathsf{E},\mathsf{F}}$ to $\mathcal{M}_{\mathsf{F}}$. We slightly abuse the notation and write it $\textcal{o}\,_{\mathsf{E},\mathsf{F}}$ as well.

This construction can be used to kill the automorphisms of the objects of $\mathcal{M}_\mathsf{F}$ as follows. Let us write ${\rm Aut}_{\mathsf{E},\mathsf{F}}(M)$ for the automorphism group of an object $M \in \mathcal{M}_{\mathsf{E},\mathsf{F}}$. Suppose that $[M] \in \mathcal{M}_\mathsf{F}$ has automorphism group ${\rm Aut}_\mathsf{F}([M])$ and that for each non-trivial $\phi \in {\rm Aut}_\mathsf{F}([M])$, $h_\phi(M) \neq M$. Then ${\rm Aut}_{\mathsf{E},\mathsf{F}}(M)$ is clearly the trivial group.  

We can ensure that all the automorphism groups are killed by taking an $\mathsf{E}$-structure to be an $\mathsf{F}$-structure on $M$ together with an element $\phi \in {\rm Aut}_\mathsf{F}([M])$. For $\phi, \psi \in {\rm Aut}_\mathsf{F}([M])$, we define
\be
h_{\phi}(\mathsf{F}([M]),\psi) = (\phi^\ast \mathsf{F}([M]), \phi \circ \psi) \;.
\ee
As the left action of ${\rm Aut}_\mathsf{F}([M])$ on itself is free, ${\rm Aut}_{\mathsf{E},\mathsf{F}}(M) = 1$ for all $M \in \mathcal{M}_{\mathsf{E},\mathsf{F}}$.

Here is an example to illustrate the constructions above.  Let $\mathsf{F}$ be the data of a conformal structure, in addition to the orientation and the smooth structure. ${\rm Aut}_\mathsf{F}([M])$ are the conformal transformations of $[M]$, i.e. the orientation preserving diffeomorphisms preserving the conformal structure. Take $\mathsf{E}$ to be the data of a Riemannian metric. Then we have a forgetful functor $\textcal{o}\,_{\mathsf{E}, \mathsf{F}}$ that associates its conformal structure to each Riemannian metric. If a Riemannian metric $g$ on $M$ has no isometries, then ${\rm Aut}_{\mathsf{E},\mathsf{F}}(M) = 1$. In the category $\mathcal{M}_{\mathsf{E},\mathsf{F}}$, morphisms correspond to conformal transformations, and do not necessarily perserve the Riemannian metrics carried by the manifolds.

We will see that in order to be able to define an associated bordism category, we cannot require the structure $\mathsf{E}$ to be smooth or even continuous. In the example above, we would allow for any Riemannian metrics, even discontinuous ones. This is of little importance, as the structure $\mathsf{E}$ is used exclusively to kill automorphisms. Of course, the underlying manifolds always carry smooth structures. 

\paragraph{The bordism 2-category} We now define the bordism 2-category $\mathcal{B}^{d,2}_{\mathsf{E},\mathsf{F}}$. The construction below makes sense when $\mathsf{E} = \mathsf{F}$, in which case we write the corresponding category $\mathcal{B}^{d,2}_{\mathsf{F}}$. In what follows, we work in the category $\mathcal{M}_{\mathsf{E},\mathsf{F}}$.

The objects of $\mathcal{B}^{d,2}_{\mathsf{E},\mathsf{F}}$ are $d-2$-dimensional closed manifolds $M^{d-2}$ endowed with a codimension 2 germ of $\mathsf{F}$-structure. Note that we do not require germs of $\mathsf{E}$-structure, for reasons explained below.

1-morphisms are of two types. First, we have \textit{regular 1-morphisms} from an object $M_-^{d-2}$ to an object $M^{d-2}_+$ that are triplets $(M^{d-1,1}, \theta^-, \theta^+)$ whose content is as follows. $M^{d-1,1}$ is a $d-1$-dimensional manifold endowed with a codimension 1 germ of $\mathsf{F}$-structure. The boundary $\partial M^{d-1,1} = \partial_- M^{d-1,1} \sqcup \partial_+ M^{d-1,1}$ is partitioned into two disjoint components. $\theta^+: M^{d-2}_+ \rightarrow \partial_+ M^{d-1,1}$ and $\theta^-: -M^{d-2}_- \rightarrow \partial_- M^{d-1,1}$ are isomorphisms. (A minus denotes the orientation flip.) A consequence of our assumptions is that the codimension 2 germ of $\mathsf{F}$-structure on $M^{d-2}_\pm$ coincides with the one obtained by pulling-back through $\theta^\pm$ the restriction of the codimension 1 germ on $M^{d-1,1}$ to $\partial_\pm M^{d-1,1}$.

For each pair $(M^{d-2}, \rho)$ composed of an object $M^{d-2}$ and an isomorphism $\rho$ with source $M^{d-2}$, we also include a 1-morphism from $M^{d-2}$ to $\rho M^{d-2}$. We call such 1-morphisms \emph{limit morphisms}. Limit morphisms can be thought of as limits as $\epsilon$ goes to zero of regular morphisms of the form $(M^{d-2} \times (-\epsilon, \epsilon), {\rm id}_{-M^{d-2}}, \rho)$. 

The composition of 1-morphisms is defined as follows for regular morphisms. Given 1-morphisms $(M^{d-1,1}_0, \theta^-_0, \theta^+_0)$ from $M^{d-2}_{0-}$ to $M^{d-2}_{0+}$ and $(M^{d-1,1}_1, \theta^-_1, \theta^+_1)$ from $M^{d-2}_{1-}$ to $M^{d-2}_{1+}$, with $M^{d-2}_{1-} = M^{d-2}_{0+}$, let $M^{d-1,1}_{01}$ be the gluing $M^{d-1,1}_{0} \sqcup_{M^{d-2}_{0+}} M^{d-1,1}_{1}$ along the maps $\theta^+_0$ and $\theta^-_1$. We define the composition to be $(M^{d-1,1}_{01}, \theta^-_0, \theta^+_1)$.  The composition involving limit morphisms is defined similarly. The composition of morphisms is strictly associative. 

Let us now explain why we do not require germs of $\mathsf{E}$-structure, and why in fact the $\mathsf{E}$-structure cannot even be required to be continuous. As we already emphasized, the morphisms of $\mathcal{M}_{\mathsf{E},\mathsf{F}}$ do not preserve $\mathsf{E}$-structures. Assuming that we would endow our manifolds with germs of $\mathsf{E}$-structures, we would not be able to consistently require these germs to coincide on the boundaries that are glued. For this reason germs cannot be used to guarantee that the gluing of smooth $\mathsf{E}$-structures yields smooth $\mathsf{E}$-structures. This is why we do not require either germs or smoothness for the $\mathsf{E}$-structures. Of course, the germs of $\mathsf{F}$-structures are required to match under gluing and all the $\mathsf{F}$-structures are smooth.

For $\rho = {\rm id}_{M^{d-2}}$, limit morphisms provide the strict identity morphisms required by the axioms of strict 2-categories. Given a manifold with $\mathsf{F}$-structure $[M^{d-2}]$, the limit morphisms between elements of $\textcal{o}_{\mathsf{E},\mathsf{F}}^{-1}([M^{d-2}])$ implement the action of the automorphism group ${\rm Aut}_\mathsf{F}([M^{d-2}])$ into the bordism category. This turns out to be very useful in relating the abstract categorical language to the physical point of view on anomalies as a symmetry breaking phenomenon (see Section \ref{SecAnFieldThe}).

Given two 1-morphisms $(M^{d-1,1}_0, \theta^-_0, \theta^+_0)$ and $(M^{d-1,1}_1, \theta^-_1, \theta^+_1)$ from $M^{d-2}_{-}$ to $M^{d-2}_{+}$, a \emph{regular 2-morphisms} from $M^{d-1,1}_0$ to $M^{d-1,1}_1$ is a pair $(M^{d,2}, \sigma)$, where $M^{d,2}$ is $d$-dimensional manifold with corners, and $\sigma$ 
\be
\label{EqBound2-Morph}
\sigma: \left( -M^{d-1,1}_0 \sqcup_{-M^{d-2}_{-} \sqcup  M^{d-2}_{+}} M^{d-1,1}_1 \right) \rightarrow \partial M^{d,2} \;,
\ee
is a isomorphism. The gluing on the left-hand side is performed using the maps $\theta^\pm_{0,1}$. 

Just like for 1-morphisms, we need to add limit 2-morphisms. Let $(M^{d-1,1}, \theta^-, \theta^+)$ be a 1-morphism from $M^{d-2}_{-}$ to $M^{d-2}_{+}$ and let $\tau$ be an isomorphism with source $(M^{d-1,1}, \theta^-, \theta^+)$ restricting to the identity on $\partial M^{d-1,1}$. Then we add a \emph{limit 2-morphism}, written $(M^{d-1,1}, \tau)$ between the source and the target of $\tau$, seen as 1-morphisms $\mathcal{B}^{d,2}_{\mathsf{E},\mathsf{F}}$.

The vertical composition of 2-morphism is defined in the obvious way, through the gluing along the relevant boundary component. The horizontal composition generates new 2-morphisms, which are chains of regular or limit 2-morphisms glued at their corners. The above defines a strict bordism 2-category $\mathcal{B}^{d,2}_{\mathsf{E},\mathsf{F}}$.

\paragraph{Truncation} We will need to consider the truncation of $\mathcal{B}^{d,2}_{\mathsf{E},\mathsf{F}}$ to manifolds and bordisms of dimension $d-1$ or lower, which we write $\mathcal{B}^{d,2}_{\mathsf{E},\mathsf{F}}|_{d-1}$. $\mathcal{B}^{d,2}_{\mathsf{E},\mathsf{F}}|_{d-1}$ is simply the 2-category that has the same objects and 1-morphisms as $\mathcal{B}^{d,2}_{\mathsf{E},\mathsf{F}}$, but whose only 2-morphisms are limit 2-morphisms. Given any 2-functor $\mathcal{F}: \mathcal{B}^{d,2}_{\mathsf{E},\mathsf{F}} \rightarrow \mathcal{C}$, for $\mathcal{C}$ some 2-category, we write $\mathcal{F}|_{d-1}: \mathcal{B}^{d,2}_{\mathsf{E},\mathsf{F}}|_{d-1} \rightarrow \mathcal{C}$ for its restriction to the truncated bordism category.

\paragraph{Pulling-back functors} Recall that we have a forgetful functor $\textcal{o}\,_{\mathsf{E},\mathsf{F}}$ from $\mathcal{M}_{\mathsf{E},\mathsf{F}}$ to $\mathcal{M}_{\mathsf{F}}$. There is an associated forgetful functor $\textcal{o}\,^\mathcal{B}_{\mathsf{E},\mathsf{F}}: \mathcal{B}^{d,2}_{\mathsf{E},\mathsf{F}} \rightarrow \mathcal{B}^{d,2}_{\mathsf{F}}$. Suppose that $\mathcal{F}$ is a functor with domain $\mathcal{B}^{d,2}_{\mathsf{F}}$. Then we can pull it back to a functor $\mathcal{F}' := \mathcal{F} \circ \textcal{o}\,^\mathcal{B}_{\mathsf{E},\mathsf{F}}$ with domain $\mathcal{B}^{d,2}_{\mathsf{E},\mathsf{F}}$. This shows that one can associate canonically a field theory functor with domain $\mathcal{B}^{d,2}_{\mathsf{E},\mathsf{F}}$ to an extended field theory defined on $\mathsf{F}$-manifolds.

\subsection{Differential cocycles}

\label{AppDiffCoc}

Let $M$ be a manifold. Let us write $C^p(M; \mathbbm{K})$ for the space of singular cochains of degree $p$ on $M$ valued in a ring $\mathbb{K}$. Let us also write $\Omega^p(M)$ for the space of real-valued differential forms of degree $p$ on $M$. An (unshifted) differential cochain of degree $p$ over $M$ is an element
\be
\check{A} = (a, h, \omega) \in C^p(M; \mathbbm{Z}) \times C^{p-1}(M; \mathbbm{R}) \times \Omega^p(M) = \check{C}(Y) \;.
\ee
We call $a$ the \emph{characteristic} of $\check{A}$ and $\omega$ its \emph{curvature}. We define a differential by
\be
d\check{A} = (da, \omega - dh - a, d\omega) \;,
\ee
where $d$ on the right-hand side denotes the differential on singular cocycles and differential forms, and we see $\omega$ as a real valued singular cocycle using integration.
The (higher) category $\check{\mathcal{Z}}^p(M)$ of differential cocycles of degree $p$ is defined as follows:
\begin{itemize}
\item Its objects are differential cocycles, i.e. differential cochains on $M$ that are closed with respect to $d$.
\item Its $1$-morphisms are differential cochains of degree $p-1$ with vanishing curvature. If $\check{B}$ is such a cochain, it provides a morphism from $\check{A}$ to $\check{A} + d\check{B}$.
\item Its $k$-morphisms are differential cochains of degree $p-k$ with vanishing curvature. If $\check{C}$ is such a cochain, it provides a $k$-morphism from the $k-1$-morphism $\check{B}$ to the $k-1$-morphism $\check{B} + d\check{C}$.
\end{itemize}
We write $\check{Z}^p(M)$ for the group of objects of $\check{\mathcal{Z}}^p(M)$. The group of isomorphism classes of objects in $\check{\mathcal{Z}}^p(M)$ is the $p$th differential cohomology group $\check{H}^p(M)$ of $M$.

\subsection{Inverse limits}

\label{AppInvLimConstr}

The inverse limit is a useful construction that we are borrowing from \cite{Freed:1994ad}.

Let $\mathcal{G}$ be a groupoid and let $\mathcal{A}$ be a functor from $\mathcal{G}$ to $\mathcal{H}_p$, for $p = 1$ or $2$. By a functor from a category into $\mathcal{H}_2$, we mean an assignment of a 2-Hilbert space to each object of $\mathcal{G}$, as well as a unitary 1-morphism of $\mathcal{H}_2$ to each morphism in $\mathcal{G}$, such that the composition of morphisms in $\mathcal{G}$ is intertwined with the composition of 1-morphisms in $\mathcal{H}_2$. Presumably this construction makes sense for more general targets but this will be sufficient for our purpose. The set of isomorphism classes of objects of $\mathcal{G}$ will be written $\bar{\mathcal{G}}$, and the isomorphism class of $G \in \mathcal{G}$ will be written $[G]$.

A \emph{section} $s$ of $\mathcal{A}$ is an assignment of an element of $\mathcal{A}(G)$ to each object $G$ of $\mathcal{G}$. If $\mathcal{G}$ is finite, sections are elements of   
\be
S := \bigoplus_{G \in \mathcal{G}} \mathcal{A}(G) \;,
\ee
where $\bigoplus$ is the usual direct sum for $p = 1$, and the additive structure on $\mathcal{H}_2$ described in Section \ref{App2Vect} for $p = 2$. We will be forced to work with groupoids that have an infinite number of elements, although the number of isomorphism classes will always be finite. In this case, the space of sections cannot be naturally pictured as an object in $\mathcal{H}_p$, as all our $p$-Hilbert spaces are finite-dimensional. Fortunately, this has no consequence on the following. Let us write $s(G) \in \mathcal{A}(G)$, for the value of the section at $G$. An \emph{invariant section} $s$ is a section satisfying the relation $s(G') = \mathcal{A}(\phi) s(G)$ for each morphism $\phi: G \rightarrow G'$ in $\mathcal{G}$. 

We define an inner product between invariant sections by
\be
\label{EqDefInProdInvSec}
(s,s') = \sum_{[G] \in \bar{\mathcal{G}}} 
 (s(G),s'(G))_{\mathcal{A}(G)} \;,
\ee
where the sum is taken over the set $\bar{\mathcal{G}}$ of isomorphism classes of objects of $\mathcal{G}$ and we used the inner product of $\mathcal{G}(G)$ on the right-hand side. As $\mathcal{A}(\phi)$ is unitary and the sections are invariant, this definition does not depend on the choice of representatives $G$ of the isomorphism classes $[G]$. The sum in \eqref{EqDefInProdInvSec} is an ordinary sum for $p = 1$, but is a direct sum for $p = 2$. 

The collection $I_\mathcal{A}$ of invariant sections of $\mathcal{A}$ is therefore an object of $\mathcal{H}_p$, the \emph{inverse limit} of $\mathcal{A}$. $I_\mathcal{A}$ decomposes into a direct sum
\be
I_\mathcal{A} = \bigoplus_{[G] \in \bar{\mathcal{G}}} I_\mathcal{A}([G])
\ee
where $I_\mathcal{A}([G]) \in \mathcal{H}_p$ is isomorphic to $\mathcal{A}(G)$. The inverse limit construction should therefore be seen as a way of assigning an object of $\mathcal{H}_p$ to each isomorphism class of objects of $\mathcal{G}$.

Note that in the case where $\mathcal{G}$ has a single object $G$, and therefore corresponds to a group $\Gamma$, $\mathcal{A}$ is an action of $\Gamma$ on a $p$-Hilbert space, and the inverse limit is given by the space of invariants of the action.

{
\small

\providecommand{\href}[2]{#2}\begingroup\raggedright\endgroup

}

\end{document}